\documentclass[10pt, a4paper, onecolumn, conference, compsocconf]{IEEEtran}

\usepackage{graphicx}
\usepackage{url}
\usepackage{times}
\usepackage{color}
\usepackage{amssymb,amsmath}
\usepackage{threeparttable}
\usepackage[ruled,vlined]{algorithm2e}

\newcommand{\Attr}{\ensuremath{\textrm{Attr}}}

\newcommand{\sys}{\mathcal{S}}
\newcommand{\pri}{\mathcal{P}}

\newcommand{\localGame}[1]{G}

\newcommand{\Reach}{\mathsf{reach}\xspace}

\newcommand{\true}{\texttt{True}}
\newcommand{\false}{\texttt{False}}

\newcommand\bool{\mathbf{B}}
\newcommand\boolform{{\cal B}}
\newcommand\eval{{\cal E}}
\newcommand{\initloc}[1]{l^0_{#1}}
\newcommand{\initeval}[1]{e^0_{#1}}
\newcommand{\initconf}{c^0}

\newcommand{\system}{\mathcal{S}}
\newcommand{\lang}{\mathcal{L}(\system)\xspace}
\newcommand{\reach}{\mathcal{R}_{\system}\xspace}
\newcommand{\conf}{\mathcal{C}_{\system}\xspace}

\newtheorem{defi}{Definition}
\newtheorem{theo}{Theorem}

\newtheorem{lemma}{Lemma}

\begin{document}
%
\title{Algorithms for Synthesizing Priorities in Component-based Systems}


\author{
\IEEEauthorblockN{Chih-Hong Cheng\authorrefmark{1},
                    Saddek Bensalem\authorrefmark{2},
                    Yu-Fang Chen\authorrefmark{3},
                    Rongjie Yan\authorrefmark{4},}
\IEEEauthorblockN{Barbara Jobstmann\authorrefmark{2},
                    Harald Ruess\authorrefmark{5},
                    Christian Buckl\authorrefmark{5},
                    Alois Knoll\authorrefmark{1}}
                    \\
\IEEEauthorblockA{\authorrefmark{1}Department of Informatics, Technische Universit\"{a}t M\"{u}nchen, Germany}
\IEEEauthorblockA{\authorrefmark{2}Verimag Laboratory, Grenoble, France}
\IEEEauthorblockA{\authorrefmark{3}Institute of Information Science, Academia Sinica, Taipei, Taiwan}
\IEEEauthorblockA{\authorrefmark{4}State Key Laboratory of Computer Science, ISCAS, Beijing, China}
\IEEEauthorblockA{\authorrefmark{5}fortiss GmbH, Germany}\\

\textbf{\url{http://www6.in.tum.de/~chengch/vissbip}}
}


\maketitle

\begin{abstract}
We present algorithms to synthesize component-based systems
that are safe and deadlock-free using priorities, which define stateless-precedence
between enabled actions. Our core method combines the concept of fault-localization (using safety-game) and fault-repair (using SAT for conflict resolution). For complex systems, we propose three complementary methods as preprocessing steps for priority synthesis, namely (a) data abstraction to reduce component complexities, (b) alphabet abstraction and $\sharp$-deadlock to ignore components, and (c) automated assumption learning for compositional priority synthesis.

\end{abstract}



%
\IEEEpeerreviewmaketitle

\section{Introduction\label{sec.algo.prioritysyn.introduction}}

Priorities~\cite{goessler2003priority} define \emph{stateless-precedence relations between actions}
available in component-based systems. They can be used to restrict the behavior of a system in order to avoid undesired states.  They are particularly useful to avoid deadlock states (i.e., states in which all actions are disabled), because they do not introduce new deadlock states and therefore avoid creating new undesired states.  Furthermore, due to their stateless property and the fact that they operate on the interface of a component, they are relatively easy to implement in a distributed setting~\cite{GrafPQ10,Bonakdarpour2011distribute}.
In a tool paper~\cite{cheng:vissbip:2011}, we presented the tool \textsc{VissBIP}\footnote{Shortcut for \underline{\textbf{Vi}}sualization and \underline{\textbf{s}}ynthesis for \underline{\textbf{s}}imple \underline{\textbf{BIP}} systems.} together with a concept called \emph{priority synthesis}, which aims to automatically generate a set of priorities such that the system constrained by the synthesized priorities satisfies a given \emph{safety property} or \emph{deadlock freedom}. In this paper, we explain the underlying algorithm and propose extensions for more complex systems. 

Priority synthesis is expensive; we showed in~\cite{cheng:hardness:2011} that synthesizing priorities for safety properties (or deadlock-freedom) is NP-complete in the size of the state space of the product graph.
%
Therefore, we present an incomplete search framework for priority synthesis, which mimics the process of \emph{fault-localization} and \emph{fault-repair} (Section~\ref{sec.algo.prioritysyn.repair}).
  Intuitively, a state is a fault location if it is the latest point from which there is a way to avoid a failure, i.e., there exists (i) an outgoing action that leads to an \emph{attracted state}, a state from which all paths unavoidably reach a bad state, and (ii) there exists an alternative action that avoids entering any of the attracted states.
%
We compute fault locations using the algorithm for \emph{safety games}.
Given a set of fault locations, priority synthesis is achieved via fault-repair: an algorithm resolves potential conflicts in priorities generated via fault-localization and finds a satisfying subset of priorities as a solution for synthesis.
Our symbolic encodings on the system, together with the new variable ordering heuristic and other optimizations, helps to solve problems much more efficiently compared to our preliminary implementation in~\cite{cheng:vissbip:2011}.
Furthermore, it allows us to integrate an adversary environment model
similar to the setting in Ramadge and Wonham's controller synthesis framework~\cite{ramadge1989control}.

Abstraction or compositional techniques are widely used in verification of infinite state or complex systems for safety properties
but \emph{not all} techniques ensure that synthesizing an abstract system for
deadlock-freeness guarantees deadlock-freeness in the concrete system (Section~\ref{sec.algo.prioritysyn.complexities}).
Therefore, it is important to find appropriate techniques to assist synthesis on complex problems.
We first revisit \emph{data abstraction} (Section~\ref{subsec.algo.prioritysyn.dp.abstraction}) for data domain such that priority synthesis works on an abstract system composed by components abstracted component-wise~\cite{bensalem2008compositional}.
Second, we present a technique called \emph{alphabet-abstraction} (Section~\ref{subsec.algo.prioritysyn.alphabet}), handling complexities induced by the composition of components.
Lastly, for behavioral-safety properties (not applicable for deadlock-avoidance),
we utilize automata-learning~\cite{angluin1987learning} to achieve \emph{compositional priority synthesis} (Section~\ref{sec.algo.prioritysyn.assume.guarantee}).

We implemented the presented algorithms (except connection with the data abstraction module in D-Finder~\cite{bensalem:dfinder2:2011}) in the \textsc{VissBIP} tool and performed experiments to evaluate them (Section~\ref{sec.algo.prioritysyn.evaluation}).  Our examples show that the process using fault-localization and fault-repair generates priorities that are highly desirable.  Alphabet abstraction enables us to scale to arbitrary large problems.  We also present a model for distributed communication.  In this example, the priorities synthesized by our engine are completely local (i.e., each priority involves two local actions within a component). Therefore, they can be translated directly to distributed control.  We summarize related work and conclude with an algorithmic flow in Section~\ref{sec.algo.prioritysyn.related} and~\ref{sec.algo.prioritysyn.conclusion}.

\section{Component-based Modeling and Priority Synthesis\label{sec.algo.prioritysyn.bip}}

\subsection{Behavioral-Interaction-Priority Framework}

The Behavior-Interaction-Priority (BIP) framework\footnote{http://www-verimag.imag.fr/Rigorous-Design-of-Component-Based.html?lang=en} provides a rigorous component-based design flow for heterogeneous systems.  Rigorous design refers to the strict separation of three different layers (behaviors, interactions, and priorities) used to describe a system.
A detailed description of the BIP language can be found in~\cite{basu2006modeling}.
To simplify the explanations, we focus on \emph{simple} systems, i.e., systems without hierarchies and finite data types.
Intuitively, a simple BIP system consists of a set of automata (extended with data) that synchronize on joint labels.

\begin{defi}[BIP System]
We define a (simple BIP) system as a tuple $\mathcal{S} = (C , \Sigma, \mathcal{P})$, where
\begin{itemize}
 \item $\Sigma$ is a finite set of \textbf{events} or interaction labels, called \textbf{interaction alphabet},
 \item $C = \bigcup_{i=1}^m C_i$ is a finite set of \textbf{components}. Each component $C_i$ is
 a transition system extended with data. Formally, $C_i$ is a tuple
$(L_i, V_i, \Sigma_i, T_i, \initloc{i}, \initeval{i})$:
 \begin{itemize}
    \item $L_i=\{l_{i_1},\ldots,l_{i_n}\}$ is a finite set of \emph{control locations}.
    \item $V_i = \{v_{i_1},\ldots,v_{i_p}\}$ is a finite set of \emph{(local) variables} with a finite domain.
    Wlog we assume that the domain is the Boolean domain $\bool = \{\true, \false\}$.
    We use $|V_i|$ to denote the number of variables used in~$C_i$.  An \emph{evaluation (or assignment)} of the variables in $V_i$ is a functions $e: V_i\to\bool$ mapping every variable to a value in the domain.  We use $\eval(V_i)$ to denote the set of all evaluations over the variables $V_i$.  Given a Boolean formula $f\in \boolform(V_i) $ over the variables in $V_i$ and an evaluation $e\in \eval(V_i)$, we use $f(e)$ to refer to the truth value of $f$ under the evaluation $e$.
    \item $\Sigma_i \subseteq \Sigma$ is a subset of interaction labels used in $C_i$.
    \item $T_i$ is the set of \emph{transitions}. A transition $t_i\in T_i$ is of the form $(l,g,\sigma,f,l')$, where $l, l'\in L_i$ are the \emph{source and destination location}, $g \in \boolform(V_i)$ is called the \emph{guard} and is a Boolean formula over the variables $V_i$.
$\sigma \in \Sigma_i$ is an interaction label (specifying the event triggering the transition), and $f: V_i \to \boolform(V_i)$ is the \emph{update function} mapping every variable to a Boolean formula
encoding the change of its value.
  \item $\initloc{i} \in L_i$ is the \emph{initial location} and $\initeval{i} \in \eval(V_i)$ is the initial evaluation of the variables.
 \end{itemize}
 \item $\mathcal{P}$ is a finite set of interaction pairs (called \textbf{priorities}) defining a relation $\prec \; \subseteq\Sigma\times\Sigma$ between the interaction labels. We require that $\prec$ is (1)
    transitive and (2) non-reflexive (i.e., there are no circular
    dependencies)~\cite{goessler2003priority}. For $(\sigma_1, \sigma_2) \in \mathcal{P}$, we sometimes write $\sigma_1 \prec \sigma_2$ to highlight the property of priority.
\end{itemize}
\end{defi}

\begin{defi}[Configuration]
Given a system $\mathcal{S}$, a \emph{configuration (or state)}~$c$ is a tuple $(l_1, e_1, \ldots, l_m, e_m)$ with $l_i \in L_i$ and $e_i \in \eval(V_i)$ for all $i \in \{1,\ldots,m\}$. We use $\conf$ to denote the set of all reachable configurations.
The configuration $(\initloc{1}, \initeval{1}, \ldots, \initloc{m}, \initeval{m})$ is called the \emph{initial configuration} of $\mathcal{S}$ and
is denoted by $\initconf$. 
\end{defi}

\newcommand{\inter}{\bar{\sigma}}

\begin{defi}[Enabled Interactions]\label{def:semantics}
Given a system $\mathcal{S}$ and a configuration $c=(l_1, e_1, \ldots, l_m, e_m)$,
we say an interaction $\sigma \in \Sigma$ is \textbf{enabled (in $c$)}, if the following conditions hold:
\begin{enumerate}
    \item (Joint participation) $\forall i \in \{1,\ldots, m\}$, if $\sigma \in \Sigma_i$, then $\exists g_i, f_i, l_i'$ such that $(l_i,g_i,\sigma,f_i,l_i') \in T_i$ and  $g_i(e_i) = \true$.
    \item (No higher priorities enabled) For all other interaction $\inter \in \Sigma$ satisfying joint participation (i.e., $\forall i \in \{1,\ldots, m\}$, if $\inter \in \Sigma_i$, then $\exists (l_i,\bar{g}_i,\inter,\bar{f}_i,\bar{l}_i') \in T_i$ such that $\bar{g}_i(e_i) = \true$), $(\sigma, \inter) \not\in \mathcal{P}$ holds.

\end{enumerate}
\end{defi}

\begin{defi}[Behavior]
  Given a system $\mathcal{S}$, two configurations $c=(l_1, e_1, \ldots,$ $l_m, e_m)$, $c'=(l'_1, e'_1, \ldots, l'_m, e'_m)$, and an interaction $\sigma\in\Sigma$ enabled in $c$, we say \emph{$c'$ is a $\sigma$-successor (configuration) of $c$}, denoted $c \xrightarrow[]{\sigma} c'$, if the following two conditions hold for all components $C_i =(L_i, V_i, \Sigma_i, T_i, \initloc{i}, \initeval{i})$:
\begin{itemize}
\item (Update for participated components) If $\sigma \in \Sigma_i$, then there exists a transition $(l_i,g_i,\sigma,f_i,l_i') \in T_i$ such that $g_i(e_i) = \true$ and for all variables $v\in V_i$, $e_i' = f_i(v)(e_i)$.
 \item (Stutter for idle components) Otherwise, $l'_i = l_i$ and $e'_i = e_i$.
\end{itemize}

Given two configurations $c$ and $c'$, we say \emph{$c'$ is reachable from $c$} with the interaction sequence $w=\sigma_1 \ldots \sigma_k$ , denoted $c \xrightarrow[]{w} c'$, if there exist configurations $c_0,\dots,c_{k}$ such that (i) $c_0 = c$,
(ii) $c_{k} = c'$, and (iii) for all  $i: 0\le i < k$, $c_i \xrightarrow[]{\sigma_{i+1}} c_{i+1}$.
We denote the set of all configuration of $\mathcal{S}$ reachable from the initial configuration $c^0$ by $\reach$.
The \emph{language} of a system $\mathcal{S}$, denoted $\lang$, is the set $\{w \in \Sigma^{*}\mid\exists c' \in \reach~\text{such that}~\initconf \xrightarrow[]{w} c'\}$.  Note that $\lang$ describes the behavior of $\mathcal{S}$, starting from the initial configuration $\initconf$.
\end{defi}

\noindent In this paper, we adapt the following simplifications:
\begin{itemize}
\item We do not consider uncontrollable events (of the environment), since the BIP language is currently not supporting them.
However, our framework would allow us to do so.  More precisely, we solve priority synthesis using a game-theoretic version of controller synthesis~\cite{ramadge1989control}, in which uncontrollability can be modeled. Furthermore, since we consider only safety properties, our algorithms can be easily adapted to handle uncontrollable events.  \item We do not consider data transfer during the interaction, as it is merely syntactic rewriting over variables between different components.  \end{itemize}

\subsection{Priority Synthesis for Safety and Deadlock Freedom\label{subsec.algo.prioritysyn.def.priority.syn}}

\begin{defi}[Risk-Configuration/Deadlock Safety]
  Given a system $\mathcal{S} = (C , \Sigma, \mathcal{P})$ and the set
  of \emph{risk configuration} $\mathcal{C}_{risk} \subseteq \conf$ (also called \emph{bad states}), the system is \textbf{safe}
  if the following conditions hold. (A system that is not safe is called \textbf{unsafe}.)

\begin{itemize}
    \item {\bf (Deadlock-free)}
$\forall c \in \reach$, $\exists
      \sigma\in\Sigma, \exists c' \in \mathcal{R}_{\mathcal{S}}: c \xrightarrow[]{\sigma} c'$
    \item {\bf (Risk-state-free)}
$\mathcal{C}_{risk} \cap \reach  = \emptyset$.
\end{itemize}
\end{defi}


\begin{defi}[Priority Synthesis]
  Given a system $\mathcal{S} = (C , \Sigma, \mathcal{P})$, and the
  set of risk configuration $\mathcal{C}_{risk} \subseteq \mathcal{C}_\sys$, priority synthesis
  searches for a set of priorities $\mathcal{P}_{+}$ such that
 \begin{itemize}
    \item For $\mathcal{P}\cup\mathcal{P}_{+}$, the defined relation $\prec_{\mathcal{P}\cup\mathcal{P}_{+}} \; \subseteq\Sigma\times\Sigma$ is also (1) transitive and (2) non-reflexive.
    \item $(C , \Sigma, \mathcal{P}\cup\mathcal{P}_{+})$ is safe.
 \end{itemize}
\end{defi}

Given a system $\system$,  we define the size of $\system$ as the size of the product graph induced by $\system$, i.e, $|\reach|+ |\Sigma|$. Then, we have the following result.

\begin{theo}[Hardness of priority synthesis~\cite{cheng:hardness:2011}]
  Given a system $\mathcal{S} = (C , \Sigma, \mathcal{P})$, finding a set
  $\mathcal{P}_{+}$ of priorities such that
  $(C , \Sigma, \mathcal{P}\cup\mathcal{P}_{+})$ is safe is NP-complete in the size of $\mathcal{S}$.
\end{theo}



We briefly mention the definition of \textbf{behavioral safety}, which is a powerful notion to capture erroneous behavioral-patterns
for the system under design.

\begin{defi}[Behavioral Safety]
  Given a system $\mathcal{S} = (C , \Sigma, \mathcal{P})$ and a regular language $\mathcal{L}_{\neg P} \subseteq \Sigma^{*}$ called the \emph{risk specification}, the system is \textbf{B-safe} if $\lang \cap \mathcal{L}_{\neg P} = \emptyset$. A system that is not B-safe is called \textbf{B-unsafe}.
\end{defi}

It is well-known that the problem of asking for behavioral safety can be reduced to the problem of risk-state freeness.
More precisely, since $\mathcal{L}_{\neg P}$ can be represented by a finite automaton $\mathcal{A}_{\neg P}$ (the monitor), priority synthesis for behavioral safety can be reduced to priority synthesis in the synchronous product of the system $\mathcal{S}$ and $\mathcal{A}_{\neg P}$ with the goal to avoid any product state that has a final state of~$\mathcal{A}_{\neg P}$ in the second component.

\section{A Framework of Priority Synthesis based on Fault-Localization and Fault-Repair\label{sec.algo.prioritysyn.repair}}

In this section, we describe our symbolic encoding scheme, followed by presenting our priority synthesis mechanism
using a fault-localization and repair approach.

\subsection{System Encoding\label{subsec.algo.prioritysyn.encoding}}

\newcommand{\enc}{enc}

Our symbolic encoding is inspired by the execution semantics of the BIP engine, which
during execution, selects one of the enabled interactions and executes the interaction.
In our engine, we mimic the process and create a two-stage transition: For each iteration,

\begin{itemize}
    \item (Stage 0) The \emph{environment} raises all enabled interactions.
    \item (Stage 1) Based on the raised interactions, the \emph{controller} selects one enabled interaction (if there exists one) while respecting the priority, and updates the state based on the enabled interaction.
\end{itemize}

Given a system $\system = (C , \Sigma, \mathcal{P})$, we use the
following sets of Boolean variables to encode~$\system$:
\begin{itemize}
    \item $\{stg, stg'\}$ is the \emph{stage indicator} and its primed
      version.
    \item $\bigcup_{\sigma \in \Sigma}\{\sigma, \sigma'\}$ are the
      variables representing interactions and their primed version.
      We use the same letter for an interaction and the
      corresponding variable, because there is a one-to-one
      correspondence between them.
    \item $\bigcup_{i = 1\ldots m} Y_i \cup Y'_i$, where
      $Y_i=\{y_{i1},\dots, y_{ik}\}$ and $Y'_i=\{y'_{i1},\dots,
      y'_{ik}\}$ are the variables and their primed version,
      respectively, used to encode the locations $L_i$. (We use a
      binary encoding, i.e., $k=\lceil log |L_i|\rceil$). Given a
      location $l \in L_i$, we use $\enc(l)$ and $\enc'(l)$ to refer
      to the encoding of $l$ using $Y_i$ and
      $Y'_i$, respectively.
   \item $\bigcup_{i = 1\ldots m}\bigcup_{v\in V_i} \{v, v'\}$ are the
     variables of the components and their primed version.
\end{itemize}

\begin{algorithm}[H]
\DontPrintSemicolon
\SetKwInOut{Input}{input}\SetKwInOut{Output}{output}
\Input{System $\mathcal{S} = (C , \Sigma, \mathcal{P})$}
\Output{Stage-0 transition predicate $\mathcal{T}_{stage_0}$}
\Begin{
    \For{$\sigma\in \Sigma$}{
     \nl  \texttt{let} predicate $P_{\sigma} := \true$\;
     }
    \For{$\sigma\in \Sigma$}{
      \For{$i = \{1,\ldots,m\}$}{
        \nl\lIf{$\sigma \in \Sigma_i$}{
            $P_{\sigma} := P_{\sigma} \wedge \bigvee_{(l,g,\sigma,f,l')\in T_i} (\enc(l) \wedge g)$\;
         }
      }
    }
    \texttt{let} predicate $\mathcal{T}_{stage_0} := stg \wedge \neg stg'$\;
    \For{$\sigma\in \Sigma$}{
    \nl $\mathcal{T}_{stage_0}:= T_{stage_0} \wedge (\sigma' \leftrightarrow P_{\sigma})$\;
    }
    \For{$i = \{1,\ldots,m\}$}{
    \nl $\mathcal{T}_{stage_0}:= T_{stage_0} \wedge \bigwedge_{y\in Y_i}  y \leftrightarrow y' \wedge \bigwedge_{v\in V_i} v \leftrightarrow v'$\;
    }
    \texttt{return} $\mathcal{T}_{stage_0}$\;
}
\caption{Generate Stage-0 transitions\label{algo.prioritysyn.plant.transition}}
\end{algorithm}

\begin{algorithm}[H]
\DontPrintSemicolon
\SetKwInOut{Input}{input}\SetKwInOut{Output}{output}
\Input{System $\mathcal{S} = (C , \Sigma, \mathcal{P})$}
\Output{Stage-1 transition predicate $\mathcal{T}_{stage_1}$}
\Begin{
    \texttt{let} predicate $\mathcal{T}_{stage_1} := \false$\;
    \For{$\sigma\in \Sigma$}{
        \texttt{let} predicate $T_{\sigma} := \neg stg \wedge stg'$\;
         \For{$i = \{1,\ldots,m\}$}{
           \If{$\sigma \in \Sigma_i$}{
          \nl     $T_{\sigma} := T_{\sigma} \wedge \bigvee_{(l,g,\sigma,f,l')\in T_i}
    (\enc(l) \wedge g \wedge \sigma \wedge \sigma'
    \wedge \enc'(l') \wedge \bigwedge_{v\in V_i} v' \leftrightarrow f(v))$\;
            }
        }
        \For{$\sigma'\in \Sigma, \sigma' \neq \sigma$}{
         \nl   $T_{\sigma} := T_{\sigma} \wedge \sigma' = \false$\;
        }
        \For{$i = \{1,\ldots,m\}$}{
           \nl \lIf{$\sigma \not\in \Sigma_i$}{
               $T_{\sigma} := T_{\sigma} \wedge \bigwedge_{y\in Y_i}  y \leftrightarrow y' \wedge \bigwedge_{v\in V_i} v \leftrightarrow v'$\;
            }
        }
        $\mathcal{T}_{stage_1} := \mathcal{T}_{stage_1} \vee T_{\sigma}$\;
    }
    \For{$\sigma_1 \prec \sigma_2 \in \mathcal{P}$}{
      \nl   $\mathcal{T}_{stage_1} := \mathcal{T}_{stage_1} \wedge
      ((\sigma_1 \wedge {\sigma_2}) \to \neg {\sigma_1}') $\;
   }

    \texttt{return} $\mathcal{T}_{stage_1}$\;
}

\caption{Generate Stage-1 transitions\label{algo.prioritysyn.control.transition}}
\end{algorithm}

We use Algorithm~\ref{algo.prioritysyn.plant.transition}
and~\ref{algo.prioritysyn.control.transition} to create transition
predicates $\mathcal{T}_{stage_0}$ and $\mathcal{T}_{stage_1}$ for
Stage~0 and~1, respectively. Note that $\mathcal{T}_{stage_0}$ and $\mathcal{T}_{stage_1}$ can be
merged but we keep them separately, in order to
(1) have an easy and direct way to synthesize priorities,
(2) allow expressing the freedom of the environment, and
(3) follow the semantics of the BIP engine.

\begin{itemize}
\item In Algorithm~\ref{algo.prioritysyn.plant.transition},  Line~2
  computes for each interaction $\sigma$ the predicate $P_{\sigma}$ representing all the
  configurations in which $\sigma$ is enabled in the current configuration.
  In Line~3, starting from the first interaction,
  $\mathcal{T}_{stage_0}$ is continuously refined by conjoining
  $\sigma' \leftrightarrow P_{\sigma}$ for each interaction
  $\sigma$, i.e., the variables $\sigma'$ is true if and only if the
  interaction $\sigma$ is enabled.
 Finally, Line~4 ensures that
the system configuration does not change in stage~$0$.
\item In Algorithm~\ref{algo.prioritysyn.control.transition},
  Line~1,~2,~3 are used to create the transition in which
  interaction~$\sigma$ is executed (Line~2 ensures that only $\sigma$
  is executed; Line~3 ensures the stuttering move of unparticipated
  components). Given a priority $\sigma_1 \prec \sigma_2$, in
  configurations in which $\sigma_1$ and $\sigma_2$ are both enabled
  (i.e., $\sigma_1 \wedge \sigma_2$ holds), the conjunction with Line~4
  removes the possibility to execute $\sigma_1$ when $\sigma_2$ is
  also available.
\end{itemize}

\clearpage
\subsection{Step A. Finding Fix Candidates using Fault-localization\label{subsec.algo.prioritysyn.findingfix}} 

\begin{figure}[t]
\centering
 \includegraphics[width=0.6\columnwidth]{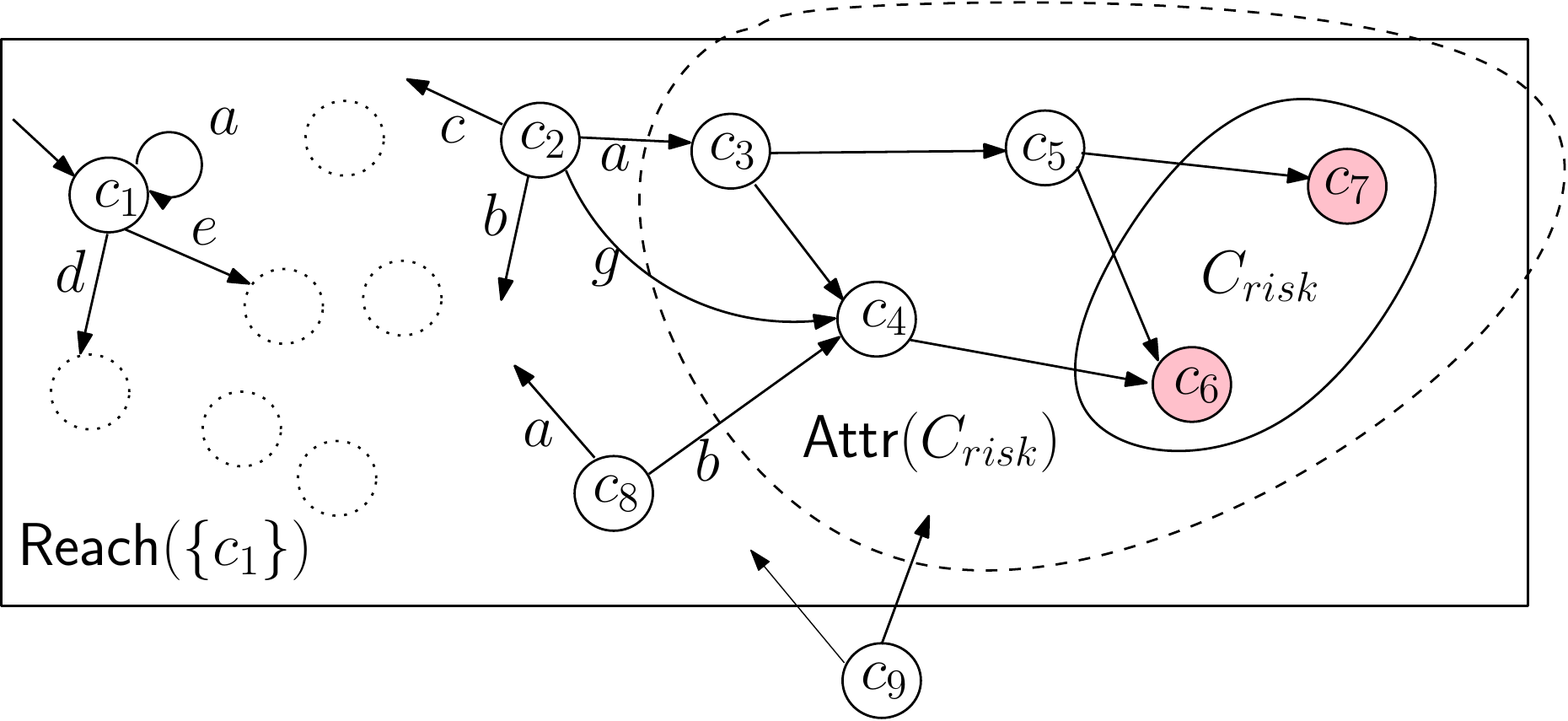}
  \caption{Locating fix candidates.}
 \label{fig:vissbip.locate.fix}
\end{figure}

Synthesizing a set of priorities to make the system safe can be done
in various ways, and we use Figure~\ref{fig:vissbip.locate.fix} to
illustrate our underlying idea. Consider a system starting from
state $c_1$. It has two risk configurations $c_6$ and $c_7$. In
order to avoid risk using priorities, one method is to work on the
initial configuration, i.e., to use the set of priorities $\{e \prec
a, d \prec a\}$.
Nevertheless, it can be
observed that the synthesized result is not very desirable, as the
behavior of the system has been greatly restricted.

Alternatively, our methodology works \emph{backwards} from the set of risk states and finds states which is able to \emph{escape from risk}. In Figure~\ref{fig:vissbip.locate.fix}, as states $c_3$, $c_4$, $c_5$ unavoidably enter a risk state, they are within the \emph{risk-attractor} ($\textsf{Attr}(\mathcal{C}_{risk})$). For state $c_2$, $c_8$, and $c_9$, there exists an interaction which avoids risk. Thus, if a set of priorities $\mathcal{P}_{+}$ can ensure that from $c_2$, $c_8$, and $c_9$, the system can not enter the attractor, then $\mathcal{P}_{+}$ is the result of synthesis. Furthermore, as $c_9$ is not within the set of reachable states from the initial configuration ($\textsf{Reach}(\{c_1\})$ in Figure~\ref{fig:vissbip.locate.fix}), then it can be eliminated without consideration. We call $\{c_2, c_8\}$ a \textbf{fault-set}, meaning that an erroneous interaction can be taken to reach the risk-attractor.

Under our formulation, we can directly utilize the result of \textbf{algorithmic game solving}~\cite{gradel:2002:automata} to compute the
fault-set. 
Algorithm~\ref{algo.fault.localization} explains the underlying computation: For conciseness, we use $\exists \Xi$ ($\exists \Xi'$) to represent existential quantification over all umprimed (primed) variables used in the system encoding. Also, we use the operator $\texttt{SUBS}(X, \Xi, \Xi')$ for variable swap (substitution) from unprimed to primed variables in $X$: the $\texttt{SUBS}$ operator
is common in most BDD packages.

\begin{itemize}
    \item In the beginning, we create $P_{ini}$ for initial configuration,  $P_{dead}$ for deadlock (no interaction is enabled), and $P_{risk}$ for risk configurations.
    \item In Part A, adding a stage-0 configuration can be computed similar to adding the environment state in a safety game. In a safety game, for an environment configuration to be added, there exists a transition which leads to the attractor.
    \item In Part A, adding a stage-1 configuration follows the intuition described earlier. In a safety game, for a control configuration~$c$ to be added, all outgoing transitions of~$c$ should lead to the attractor. This is captured by the set difference operation $\texttt{PointTo} \setminus \texttt{Escape}$ in Line~5.
    \item In Part B, Line~7 creates the transition predicate entering the attractor. Line~8 creates predicate $\texttt{OutsideAttr}$ representing the set of stage-1 configuration outside the attractor. In Line~9, by conjuncting with $\texttt{OutsideAttr}$ we ensure that the algorithm does not return a transition within the attractor.
    \item Part C removes transitions whose source is not within the set of reachable states.
\end{itemize}

\newcommand\eq{\leftrightarrow}

\begin{algorithm}[H]
\DontPrintSemicolon
\SetKwInOut{Input}{input}\SetKwInOut{Output}{output}
\Input{System $\mathcal{S} = (C , \Sigma, \mathcal{P})$, $\mathcal{T}_{stage_0}$, $\mathcal{T}_{stage_1}$}
\Output{$\mathcal{T}_{f} \subseteq \mathcal{T}_{stage_1}$ as the set of stage-1 transitions starting from the fault-set but entering the risk attractor}
\Begin{
    \textbf{let} $P_{ini} := stg \:\wedge \bigwedge_{i = 1\ldots m}  (\enc(\initloc{i}) \wedge \bigwedge_{v\in V_i} v \eq \initeval{i}(v))$\;
    \textbf{let} $P_{dead} := \neg stg \wedge \bigwedge_{\sigma \in
      \Sigma} \neg {\sigma} $\;
    \textbf{let} $P_{risk} := \neg stg \wedge \bigvee_{(l_{1}, e_{1}, \ldots, l_{m}, e_{m})\in \mathcal{C}_{risk}}$
      $(\enc(l_1) \wedge \bigwedge_{v\in V_1} v\eq e_1(v) \wedge \ldots$
      $\enc(l_m) \wedge \bigwedge_{v\in V_m} v\eq e_m(v)) $\;
    \;
    \tcp{Part A: solve safety game}
    \textbf{let} $\Attr_{pre} := P_{dead} \vee P_{risk}$,  $\Attr_{post} := \false$\;
    \nl \While{\true}{
    \tcp{add stage-0 (environment) configurations}
    \nl $\Attr_{post,0} := \exists\Xi': (\mathcal{T}_{stage_0} \wedge \texttt{SUBS}((\exists \Xi': \Attr_{pre}), \Xi, \Xi'))$\;
    \tcp{add stage-1 (system) configurations}
    \nl \textbf{let} $\texttt{PointTo} := \exists\Xi': (\mathcal{T}_{stage_1} \wedge \texttt{SUBS}((\exists \Xi': \Attr_{pre}), \Xi, \Xi'))$\;
    \nl \textbf{let} $\texttt{Escape} := \exists\Xi': (\mathcal{T}_{stage_1} \wedge \texttt{SUBS}((\exists \Xi': \neg \Attr_{pre}), \Xi, \Xi'))$\;
    \nl $\Attr_{post,1} := \texttt{PointTo} \setminus \texttt{Escape}$\;
    \nl $\Attr_{post} := \Attr_{pre} \vee \Attr_{post,0} \vee \Attr_{post,1} $\tcp*[r]{Union the result}
     \lIf{$\Attr_{pre} \leftrightarrow \Attr_{post}$}{
        \texttt{break}\tcp*[r]{Break when the image saturates}
     }
     \lElse{
     $\Attr_{pre} := \Attr_{post}$\;
     }
    }
    \;
    \tcp{Part B: extract $\mathcal{T}_{f}$ }
    \nl $\texttt{PointTo} := \mathcal{T}_{stage_1} \wedge \texttt{SUBS}((\exists \Xi': \Attr_{pre}), \Xi, \Xi'))$\;
    \nl $\texttt{OutsideAttr} := \neg \Attr_{pre} \wedge (\exists \Xi': \mathcal{T}_{stage_1})$\;
    \nl $\mathcal{T}_{f} := \texttt{PointTo} \wedge \texttt{OutsideAttr}$\;
    \;
    \tcp{Part C: eliminate unused transition using reachable states}
    \textbf{let} $\Reach_{pre}:= P_{ini}$, $\Reach_{post}:= \false$\;
    \nl \While{\true}{
      $\Reach_{post} := \Reach_{pre} \vee \texttt{SUBS}(\exists \Xi: (\Reach_{pre} \wedge (\mathcal{T}_{stage_0} \vee \mathcal{T}_{stage_1})), \Xi', \Xi)$\;
      \lIf{$\Reach_{pre} \leftrightarrow \Reach_{post}$}{
            \texttt{break}\tcp*[r]{Break when the image saturates}
      }
      \lElse{$\Reach_{pre} := \Reach_{post}$\;
      }
    }
    \nl \textbf{return} $\mathcal{T}_{f} \wedge \Reach_{post}$
}
\caption{Fault-localization\label{algo.fault.localization}}
\end{algorithm}

\subsection{Step B. Priority Synthesis via Conflict Resolution \label{subsec.algo.prioritysyn.repair} - from Stateful to Stateless} 

Due to our system encoding, in
Algorithm~\ref{algo.fault.localization}, the return value
$\mathcal{T}_{f}$ contains not only the risk interaction but also all
possible interactions simultaneously available. Recall Figure~\ref{fig:vissbip.locate.fix},
$\mathcal{T}_{f}$ returns three transitions, and we can extract \textbf{priority candidates} from each transition.
\begin{itemize}
    \item On $c_2$, $a$ enters the risk-attractor, while $b,g,c$ are also available. We have the following candidates $\{a \prec b, a\prec g, a \prec c\}$.
    \item On $c_2$, $g$ enters the risk-attractor, while $a,b,c$ are also available. We have the following candidates $\{g \prec b, g\prec c, g \prec a\}$\footnote{Notice that at least one candidate is a true candidate for risk-escape. Otherwise, during the attractor computation, $c_2$ will be included within the attractor.}.
    \item On $c_8$, $b$ enters the risk-attractor, while $a$ is also available. We have the following candidate $b \prec a$.
\end{itemize}

From these candidates, we can perform \textbf{conflict resolution} and
generate a set of priorities that ensures avoiding the attractor.  For
example, $\{a\prec c, g \prec a, b \prec a\}$ is a set of satisfying
priorities to ensure safety. Note that the set $\{a\prec b, g \prec b,
b \prec a\}$ is not a legal priority set, because it creates circular
dependencies. In our implementation, conflict resolution is performed
using SAT solvers: In the SAT problem, any priority $\sigma_1 \prec
\sigma_2$ is presented as a Boolean variable $\underline{\sigma_1
  \prec \sigma_2}$, which can be set to $\true$ or $\false$.  If the
generated SAT problem is satisfiable, for all variables
$\underline{\sigma_1 \prec \sigma_2}$ which is evaluated to $\true$,
we add priority $\sigma_1 \prec \sigma_2$ to $\mathcal{P}_{+}$. The
synthesis engine creates four types of clauses. 
\begin{enumerate}
    \item \textbf{[Priority candidates]} For each edge $t \in \mathcal{T}_{f}$ which enters the risk attractor using $\sigma$ and having $ \sigma_1, \ldots, \sigma_e$ available actions (excluding $\sigma$), create clause $(\bigvee_{i=1 \ldots e} \underline{\sigma \prec \sigma_{i}})$\footnote{In implementation, Algorithm~\ref{algo.fault.localization} works symbolically on BDDs and proceeds on \textbf{cubes} of the risk-edges (a cube contains a set of states having the same enabled interactions and the same risk interaction), hence it avoids enumerating edges state-by-state.}.
    \item \textbf{[Existing priorities]} For each priority $\sigma \prec \sigma' \in \mathcal{P}$, create
        clause $(\underline{\sigma \prec \sigma'})$.
    \item \textbf{[Non-reflective]} For each interaction $\sigma$ used in (1) and (2), create clause
            $(\neg \underline{\sigma \prec \sigma})$.
    \item \textbf{[Transitive]} For any three interactions $\sigma_1, \sigma_2, \sigma_3$  used in (1) and (2), create clause
            $((\underline{\sigma_1 \prec \sigma_2} \wedge \underline{\sigma_2 \prec \sigma_3})\Rightarrow \underline{\sigma_1 \prec \sigma_3})$.
\end{enumerate}
\noindent When the problem is satisfiable, we only output the set of priorities within the priority candidates (as non-reflective and transitive clauses are inferred properties). Admittedly, here we still solve an NP-complete problem. Nevertheless,
\begin{itemize}
    \item The number of interactions involved in the fault-set can be much smaller than $\Sigma$.
    \item As the translation does not involve complicated encoding, we observe from our experiment that solving the SAT problem does not occupy a large portion (less than $20\%$ for all benchmarks) of the total execution time.
\end{itemize}

%


\subsection{Optimization}

\begin{figure}[t]
\centering
 \includegraphics[width=0.4\columnwidth]{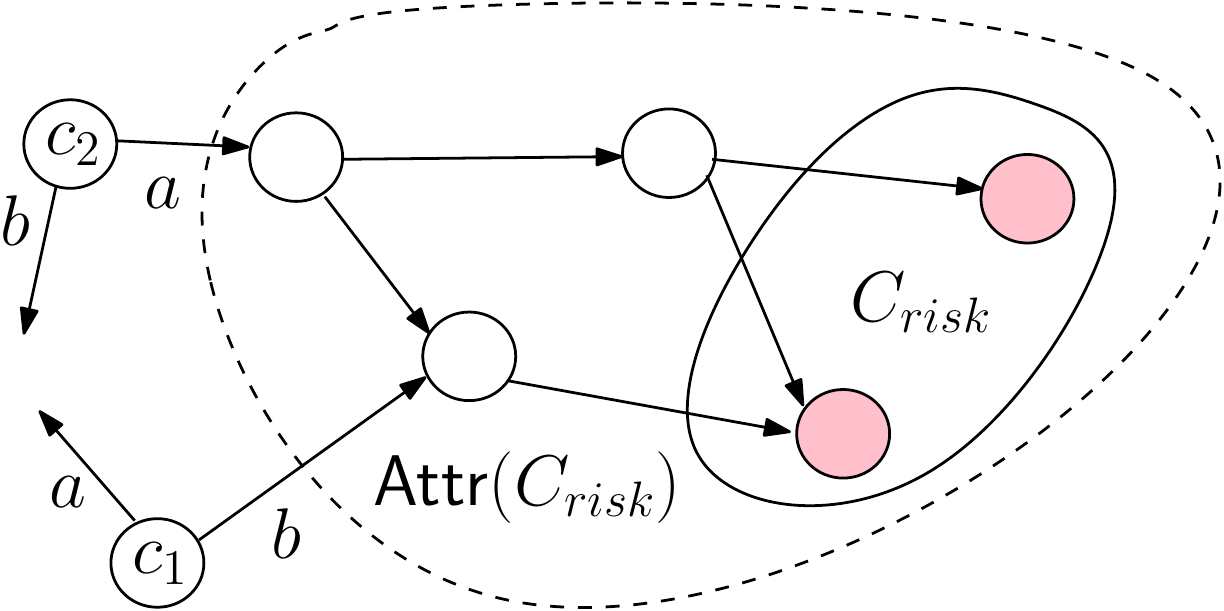}
  \caption{A simple scenario where conflicts are unavoidable on the fault-set.}
 \label{fig:vissbip.repushing}
\end{figure}

Currently, we use the following optimization techniques compared to the preliminary implementation of~\cite{cheng:vissbip:2011}.

\subsubsection{Handling unsatisfiability}
In the resolution scheme in Section~\ref{subsec.algo.prioritysyn.repair}, when the generated SAT problem is unsatisfiable, we can redo the process by moving some states in the fault-set to the attractor. This procedure is implemented by selecting a subset of priority candidates and annotate to the original system. We call this process \textbf{priority-repushing}. E.g., consider the system $\mathcal{S} = (C , \Sigma, \mathcal{P})$ in Figure~\ref{fig:vissbip.repushing}. The fault-set $\{c_1, c_2\}$ is unable to resolve the conflict: For $c_1$ the priority candidate is $a \prec b$, and for $c_2$ the priority candidate is $b \prec a$. When we redo the analysis with $\mathcal{S} = (C , \Sigma, \mathcal{P} \cup \{a \prec b\})$, this time $c_2$ will be in the attractor, as now $c_2$ must respect the priority and is unable to escape using $a$. Currently in our implementation, we supports the repushing under fixed depth to increase the possibility of finding a fix.

\subsubsection{Initial Variable Ordering: Modified FORCE Heuristics}
As we use BDDs to compute the risk-attractor, a good initial variable ordering can greatly influence the total required time solving the game. Although finding an optimal initial variable ordering is known to be NP-complete~\cite{tani1993complexity}, many heuristics can be applied to find a good yet non-optimal ordering\footnote{Also, dynamic variable ordering, a technique which changes the variable ordering at run-time, can be beneficial when no good variable ordering is known~\cite{clarke1999mc}}. The basic idea of these heuristics is to group variables close if they participate in the same transition~\cite{clarke1999mc}; experiences have shown that this creates a BDD diagram of smaller size. Thus our goal is to find a heuristic algorithm which can be computed efficiently while creates a good ordering.

We adapt the concept in the FORCE heuristic~\cite{aloul2003force}. Although the purpose of the FORCE heuristic is to work on SAT problems, we find the concept very beneficial in our problem setting. We explain the concept of FORCE based on the example in~\cite{aloul2003force}, and refer interested readers to the paper~\cite{aloul2003force} for full details.

Given a CNF formula $C=c_1\wedge c_2 \wedge c_3$, where $c_1=(a\vee c), c_2=(a\vee d), c_3=(b\vee d)$.
\begin{itemize}
    \item Consider a variable ordering $\langle a,b,c,d\rangle$. For this ordering, we try to evaluate it by considering the sum of the \emph{span}. A span is the maximum distance between any two variables within the same clause. For $c_1$, under the ordering the span equals~$2$; for~$c_2$ the span equals~$3$, and the sum of the span equals~$7$.
    \item Consider another variable ordering $\langle c,a,d,b\rangle$. Then the sum of span equals~$3$. Thus we consider that $\langle c,a,d,b\rangle$ is superior than $\langle a,b,c,d\rangle$.
    \item The purpose of the FORCE heuristic is to reduce the sum of such span. In the CNF example, the name of the heuristics suggests that a conceptual force representing each clause is grouping variables used within the clause.
\end{itemize}

Back to priority synthesis, consider the set of components $\bigcup_{i=1}^{n} C_i$ together with interaction labels $\Sigma$.
We may similarly compute the sum of all spans, where now a span is \emph{the maximum distance between any two components
participating the same interaction $\sigma\in \Sigma$}. Precisely, we analogize clauses and variables in the original FORCE heuristic with interaction symbols and components. Therefore, we regard the FORCE heuristics equally applicable to create a
better initial variable ordering for priority synthesis.

\noindent\textbf{[Algorithm Sketch]} Our modified FORCE heuristics is as follows.
\begin{enumerate}
    \item Create an initial order of vertices composed from a set of components $\bigcup_{i=1}^{n} C_i$ and interactions $\sigma\in \Sigma$. Here we allow the user to provide an initial variable ordering, such that the FORCE heuristic can be applied more efficiently.
    \item Repeat for limited time or until the span stops decreasing:
        \begin{itemize}
            \item Create an empty list.
            \item For each interaction label $\sigma\in \Sigma$, derive its center of gravity $COG(\sigma)$ by computing the \emph{average position} of all participated components. Use the average position as its value. Add the interaction with the value to the list.
            \item For each component $C_i$, compute its value by $\frac{\sum_{\sigma \in Sigma_i}COG(\sigma)}{|\Sigma_i|}$. Add the component with the value to the list.
            \item Sort the list based on the value. The resulting list is considered as a new variable ordering. Compute the new span and compare with the span from the previous ordering.
        \end{itemize}
\end{enumerate}

\subsubsection{Dense variable encoding} The encoding in Section~\ref{subsec.algo.prioritysyn.encoding} is \emph{dense} compared to the encoding in~\cite{cheng:vissbip:2011}. In~\cite{cheng:vissbip:2011}, for each component $C_i$ participating interaction $\sigma$, one separate variable $\sigma_i$ is used. Then a joint action is done by an \texttt{AND} operation over all variables, i.e., $\bigwedge_i \sigma_i$. This eases the construction process but makes BDD-based game solving very inefficient: For a system $\mathcal{S}$, let $\Sigma_{use1}\subseteq \Sigma$ be the set of interactions where only one component participates within. Then the encoding in~\cite{cheng:vissbip:2011} uses at least $2|\Sigma\setminus \Sigma_{use1}|$ more BDD variables than the dense encoding.

\subsubsection{Safety Engine Speedup} Lastly, as our created game graph is \emph{bipartite}, Algorithm~\ref{algo.fault.localization} can be refined to work on two separate images of stage-0 and stage-1, such that line~2 and line~\{3,4\} are executed in alternation.

\section{Handling Complexities\label{sec.algo.prioritysyn.complexities}}

In verification, it is standard to use \emph{abstraction} and \emph{modularity} to reduce the complexity of the analyzed systems.
Abstraction is also useful in synthesis.
However, note that if an abstract system is deadlock-free, it does not imply that the concrete system is as well.
E.g., in Figure~\ref{fig:vissbip.abstraction}, the system composed by $C_1$ and $C_2$ contains deadlock (if both interactions $a$ and $b$ are required to be paired for execution). However, when we over-approximate $C_1$ to an abstract system $C_1^{\alpha}$, a system composed by $C_1^{\alpha}$ and $C_2$ is deadlock free. On the other hand, deadlock-freeness of an under-approximation also does not imply deadlock-freeness of a concrete system. An obvious example can be obtained by under-approximating the system $C_1$ in Figure~\ref{fig:vissbip.abstraction} to an abstract system $C_1^{\beta}$. Again, the composition of $C_1^{\beta}$ and $C_2$ is deadlock-free, while the concrete system is not.
Therefore, it is challenging to find a suitable abstract system such that the abstract system is deadlock-free implying that the concrete system is also deadlock-free.

\begin{figure}[h]
\centering
 \includegraphics[width=0.45\columnwidth]{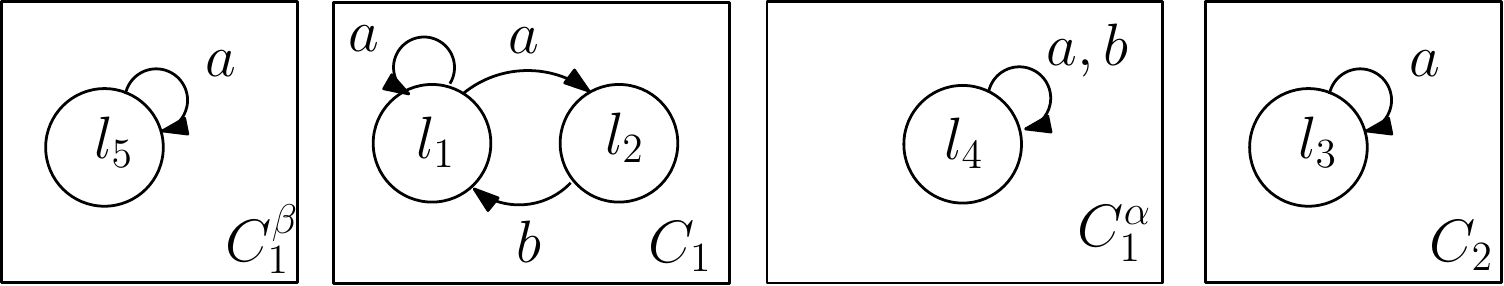}
  \caption{A scenario where the concrete system contains deadlock, but the abstract system is deadlock free.}
 \label{fig:vissbip.abstraction}
\end{figure}

In the following, we propose three techniques.

\subsection{Data abstraction\label{subsec.algo.prioritysyn.dp.abstraction}}

Data abstraction techniques presented in the previous work~\cite{bensalem2008compositional} and implemented in the D-Finder tool kit~\cite{bensalem:dfinder2:2011} are \emph{deadlock preserving}, i.e., synthesizing the abstract system to be deadlock free ensures that the concrete system is also deadlock free. Basically, the method works on an abstract system composed by components abstracted component-wise from concrete components. For example, if an abstraction preserves all control variables (i.e., all control variables are mapped by identity) and the mapping between the concrete and abstract system is precise with respect to all guards and updates (for control variables) on all transitions, then it is deadlock preserving. For further details, we refer interested readers to~\cite{bensalem2008compositional,bensalem:dfinder2:2011}.

\subsection{Alphabet abstraction\label{subsec.algo.prioritysyn.alphabet}}

\begin{figure}
\centering
 \includegraphics[width=0.5\columnwidth]{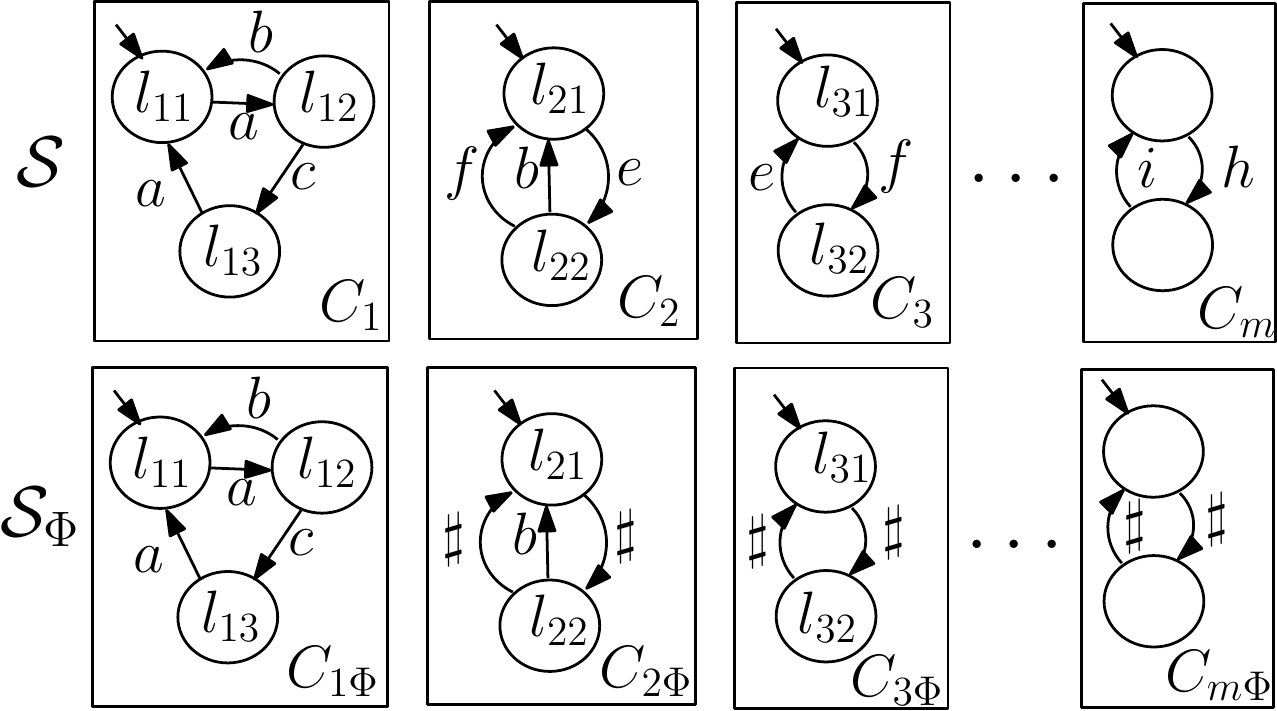}
  \caption{A system $\mathcal{S}$ and its $\sharp$-abstract system $\sys_{\Phi}$, where $\Sigma_{\Phi} = \Sigma \setminus \{a,b,c\}$.}
 \label{fig:vissbip.alphabet.abstraction}
\end{figure}

Second, we present \emph{alphabet abstraction}, targeting to synthesize priorities to avoid deadlock
(but also applicable for risk-freeness with extensions). The underlying intuition is to abstract concrete behavior of
components out of concern.

\begin{defi}[Alphabet Transformer]
Given a set $\Sigma$ of interaction alphabet. Let $\Sigma_{\Phi} \subseteq \Sigma$ be \textbf{abstract alphabet}.
Define $\alpha: \Sigma \rightarrow (\Sigma\setminus\Sigma_{\Phi})\cup\{\sharp\}$ as the alphabet transformer, such that
for $\sigma \in \Sigma$,
\begin{itemize}
    \item If $\sigma \in \Sigma_{\Phi}$, then $\alpha(\sigma) := \sharp$.
    \item Otherwise, $\alpha(\sigma) := \sigma$.
\end{itemize}
\end{defi}

\begin{defi}[Alphabet Abstraction: Syntax]
Given a system $\sys=(C,\Sigma,\pri)$ and abstract alphabet $\Sigma_{\Phi}\subseteq \Sigma$, define the \textbf{$\sharp$-abstract system} $\sys_{\Phi}$  to be $(C_{\Phi},(\Sigma\setminus\Sigma_{\Phi})\cup\{\sharp\}, \pri_{\Phi})$, where
\begin{itemize}
    \item $C_{\Phi} = \bigcup_{i=1\ldots m} C_{i\Phi}$, where $C_{i\Phi}= (L_i, V_i, \Sigma_{i\Phi},  T_{i\Phi}, \initloc{i}, \initeval{i})$ changes from $C_{i}$ by \textbf{syntactically} replacing every occurrence of $\sigma \in \Sigma_i$ to $\alpha(\sigma)$.
    \item $\pri= \bigcup_{i=1\ldots k} \sigma_i\prec\sigma_i'$ changes to $\pri_{\Phi} = \bigcup_{i=1\ldots k} \alpha(\sigma_i)\prec\alpha(\sigma_i')$, and the relation defined by $\pri_{\Phi}$ should be transitive and nonreflexive.
\end{itemize}
\end{defi}

The definition for a configuration (state) of a $\sharp$-abstract system follows Definition~2.
Denote the set of all configuration of $\mathcal{S}_{\Phi}$ reachable from $c_0$ as $\mathcal{C}_{\mathcal{S}_{\Phi}}$.
The update of configuration for an interaction $\sigma \in\Sigma\setminus\Sigma_{\Phi}$ follows Definition~3.
The only difference is within the semantics of the $\sharp$-interaction.

\begin{defi}[Alphabet Abstraction: Semantics for $\sharp$-interaction]\label{def:sharp.semantics}
Given a configuration $c=(l_1, v_1, \ldots, l_m, v_m)$, the $\sharp$-interaction is \textbf{enabled} if the following conditions hold.
\begin{enumerate}
    \item ($\geq 1$ participants) \textbf{Exists} $i \in \{1,\ldots, m\}$ where $\sharp \in \Sigma_{i\Phi}$, $\exists t_i = (l_i,g_i,\sharp,f_i,l_i') \in T_{i\Phi}$ such that $g(v_i) = \texttt{True}$.
    \item (No higher priorities enabled) There exists no other interaction $\sigma_{\flat} \in \Sigma, (\sharp, \sigma_{\flat}) \in \pri_{\Phi}$ such that $\forall i \in \{1,\ldots, m\}$ where $\sigma_{\flat} \in \Sigma_i$, $\exists t_{i\flat} = (l_i ,g_{i\flat},\sigma_{i\flat},f_{i\flat},l_i'') \in T_i$, $g_{i\flat}(v_i) = \texttt{True}$.
\end{enumerate}
Then for a configuration $c=(l_1, v_1, \ldots, l_m, v_m)$, the configuration after taking an enabled $\sharp$-interaction changes to
$c^{\flat} = (l_1^{\flat}, v_1^{\flat}, \ldots, l_m^{\flat}, v_m^{\flat})$:
\begin{itemize}
    \item (\textbf{May-update} for participated components) If $\sharp \in \Sigma_i$, then for transition $t_i = (l_i,g_i,\sharp,f_i,l_i') \in T_{i\Phi}$ such that $g_i(v_i) = \texttt{True}$, either
        \begin{enumerate}
        \item $l_i^{\flat} = l_i'$, $v_i^{\flat} = f_i(v_i)$, or
        \item $l_i^{\flat} = l_i$, $v_i^{\flat} = v_i$.
        \end{enumerate}
        Furthermore, at least one component updates (i.e., select option~1).
    \item (Stutter for unparticipated components) If $\sharp \not\in \Sigma_i$, $l_i^{\flat} = l_i$, $v_i^{\flat} = v_i$.
\end{itemize}
\end{defi}

\noindent Lastly, the behavior of a $\sharp$-abstract system follows Definition~4. In summary, the above definitions indicate that in a $\sharp$-abstract system, any local transitions having alphabet symbols within $\Sigma_{\Phi}$ can be executed in isolation or jointly. Thus, we have the following result.

\begin{lemma}\label{lem:simulation}
Given a system $\sys$ and its $\sharp$-abstract system $\sys_{\Phi}$,
define $\mathcal{R}_{\sys}$ ($\mathcal{R}_{\sys_{\Phi}}$) be the reachable states of system $\sys$ (corresponding $\sharp$-abstract system)
from from the initial configuration $c^0$. Then $\mathcal{R}_{\sys}\subseteq \mathcal{R}_{\sys_{\Phi}}$.
\end{lemma}

\begin{proof}
Result from the comparison between Definition~\ref{def:semantics} and~\ref{def:sharp.semantics}.
\end{proof}

As alphabet abstraction looses the execution condition by overlooking paired interactions,
a $\sharp$-abstract system is deadlock-free does not imply that the concrete system is deadlock free.
E.g., consider a system $\mathcal{S}'$ composed only by $C_2$ and $C_3$ in Figure~\ref{fig:vissbip.alphabet.abstraction}.
When $\Phi = \Sigma \setminus \{b\}$, its $\sharp$-abstract system $\mathcal{S}_{\Phi}'$ is shown below.
In $\mathcal{S'}$, when $C_2$ is at location $l_{21}$ and $C_3$ is at location $l_{31}$, interaction $e$ and $f$ are disabled, meaning that there exists a deadlock from the initial configuration. Nevertheless, in $\mathcal{S}_{\Phi}'$, as the $\sharp$-interaction is always enabled, it is deadlock free.

In the following, we strengthen the deadlock condition by the notion of \textbf{$\sharp$-deadlock}. Intuitively,
a configuration is $\sharp$-deadlocked, if it is deadlocked, or the only interaction available is the $\sharp$-interaction.

\begin{defi}[$\sharp$-deadlock]\label{def:sharplock}
Given a $\sharp$-abstract system $\sys_{\Phi}$, a configuration $c\in \mathcal{C}_{\mathcal{S}_{\Phi}}$ is $\sharp$-deadlocked, if
 $ \nexists \sigma \in \Sigma\setminus\Sigma_{\Phi}, c' \in \mathcal{C}_{\mathcal{S}_{\Phi}}$ such that $c\xrightarrow{\sigma} c'$.
\end{defi}
In other words, a configuration $c$ of $\sys_{\Phi}$ is $\sharp$-deadlocked implies that all interactions labeled with $\Sigma\setminus\Sigma_{\Phi}$ are disabled at $c$.

\begin{lemma}\label{lem:deadlock}
Given a system $\sys$ and its $\sharp$-abstract system $\sys_{\Phi}$, define ${\cal{D}}$ as the set of deadlock states reachable from the initial state in $\sys$, and ${\cal{D}}^\sharp$ as the set of $\sharp$-deadlock states reachable from the initial state in $\sys_{\Phi}$.
Then ${\cal{D}} \subseteq {\cal{D}}^\sharp$.
\end{lemma}

\begin{proof}
Consider a deadlock state $c\in {\cal{D}}$.
\begin{enumerate}
    \item Based on Lemma~\ref{lem:simulation}, $c$ is also in $\mathcal{R}_{\sys_{\Phi}}$.
    \item In $\sys$, as $c \in {\cal{D}}$, all interactions are disabled in $c$. Then correspondingly in $\sys_{\Phi}$, for state $c$, any interaction $\sigma\in \Sigma\setminus \Sigma_{\Phi}$ is also disabled. Therefore, $c$ is $\sharp$-deadlocked.
\end{enumerate}
Based on~1 and~2, $c \in {\cal{D}}^\sharp$. Thus ${\cal{D}} \subseteq {\cal{D}}^\sharp$.
\end{proof}

\begin{theo}\label{the:free}
Given a system $\sys$ and its $\sharp$-abstract system $\sys_{\Phi}$, if $\sys_{\Phi}$ is $\sharp$-deadlock-free, then $\sys$ is deadlock-free.
\end{theo}

\begin{proof}
As $\sys_{\Phi}$ is $\sharp$-deadlock-free, we have $\mathcal{R}_{\sys_{\Phi}}\cap {\cal{D}}^\sharp=\emptyset$. According to Lemma \ref{lem:simulation} and \ref{lem:deadlock}, we have $\mathcal{R}_{\sys}\subseteq \mathcal{R}_{\sys_{\Phi}}$ and ${\cal{D}}\subseteq {\cal{D}}^\sharp$. Hence $\mathcal{R}_{\sys}\cap {\cal{D}}=\emptyset$, implying that $\sys$ is deadlock-free.
\end{proof}

\noindent{(Algorithmic issues)} Based on the above results, the use of alphabet abstraction and the notion of $\sharp$-deadlock offers a methodology for priority synthesis working on abstraction. Detailed steps are presented as follows.
 \begin{enumerate}
 \item Given a system $\sys$, create its $\sharp$-abstract system $\sys_{\Phi}$ by a user-defined $\Sigma_{\Phi} \subseteq \Sigma$. In our implementation, we let users select a subset of components $C_{s_1},\ldots, C_{s_k}\in C$, and generate $\Sigma_{\Phi} = \Sigma\setminus(\Sigma_{s_1}\cup\ldots\cup\Sigma_{s_k})$.
     \begin{itemize}
        \item E.g., consider system $\mathcal{S}$ in Figure~\ref{fig:vissbip.alphabet.abstraction} and its $\sharp$-abstract system $\sys_{\Phi}$. The abstraction is done by looking at $C_{1}$ and maintaining $\Sigma_1 =\{a,b,c\}$.
        \item When a system contains no variables, the algorithm proceeds by eliminateing components whose interaction are completely in the abstract alphabet. In Figure~\ref{fig:vissbip.alphabet.abstraction}, as for $i= \{3\ldots m\}$, $\Sigma_{i\Phi} = \{\sharp\}$, it is sufficient to eliminate all of them during the system encoding process.
    \end{itemize}
\item  If $\sys_{\Phi}$ contains $\sharp$-deadlock states, we could obtain a $\sharp$-deadlock-free system by synthesizing a set of priorities $\pri_{+}$, where the defined relation $\prec_{+} \subseteq ((\Sigma\setminus\Sigma_{\Phi})\cup \{\sharp\})\times (\Sigma\setminus\Sigma_{\Phi})$ using techniques presented in Section~\ref{sec.algo.prioritysyn.repair}.
    \begin{itemize}
        \item In the system encoding, the predicate $P_{\sharp dead}$ for $\sharp$-deadlock is defined as $stg=\false \wedge \bigwedge_{\sigma \in \Sigma\setminus\Sigma_{\Phi}} \sigma =\false$.
        \item If the synthesized priority is having the form $\sharp \prec \sigma$, then translate it into a set of priorities $\bigcup_{\sigma' \in \Sigma_{\Phi}} \sigma' \prec \sigma$.
    \end{itemize}

\end{enumerate}

\section{Assume-guarantee Based Priority Synthesis\label{sec.algo.prioritysyn.assume.guarantee}}

We use an assume-guarantee based compositional synthesis algorithm for behavior safety. Given a system $\mathcal{S} = (C_1 \cup C_2 , \Sigma, \mathcal{P})$ and a risk specification described by a \textit{deterministic finite state automaton} $R$, where $\mathcal{L}(R)\subseteq \Sigma^*$. We use $|\mathcal{S}|$ to denote the size of $\mathcal{S}$ and $|R|$ to denote the number of states of $R$.
The synthesis task is to find a set of priority rules $\mathcal{P}_{+}$ such that adding $\mathcal{P}_{+}$ to the system $\mathcal{S}$ can make it B-Safe with respect to the risk specification $\mathcal{L}(R)$. This can be done using an \emph{assume-guarantee} rule that we will describe in the next paragraph.

We first define some notations needed for the rule. The system $\mathcal{S}_{+} = (C_1 \cup C_2 , \Sigma, \mathcal{P} \cup \mathcal{P}_{+})$ is obtained by adding priority rules $\mathcal{P}_{+}$ to the system $\mathcal{S}$.
We use $\mathcal{S}_1 = (C_1 , \Sigma, \mathcal{P}\cap \Sigma\times\Sigma_1)$ and $\mathcal{S}_2 = (C_2 , \Sigma, \mathcal{P}\cap \Sigma\times\Sigma_2)$ to denote two sub-systems of $\mathcal{S}$. We further partition the alphabet $\Sigma$ into three parts $\Sigma_{12}$, $\Sigma_1$, and $\Sigma_2$, where $\Sigma_{12}$ is the set of interactions appear both in the sets of components $C_1$ and $C_2$ (in words, the shared alphabet of $C_1$ and $C_2$), $\Sigma_{i}$ is the set of interactions appear only in the set of components $C_i$ (in words, the local alphabet of $C_i$) for $i=1,2$. Also, we require that the decomposition of the system must satisfy that $\mathcal{P} \subseteq \Sigma\times(\Sigma_1 \cup \Sigma_2)$, which means that we do not allow a shared interaction to have a higher priority than any other interaction.
This is \textbf{required} for the soundness proof of the assume-guarantee rule, as we also explained later
that we will \textbf{immediately lose soundness by relaxing this restriction}.
For $i=1,2$, the system $\mathcal{S}_{i+}=(C_i \cup \{d_i\} , \Sigma, (\mathcal{P}\cap \Sigma\times\Sigma_i) \cup \mathcal{P}_i)$ is obtained by (1) adding priority rules $\mathcal{P}_{i}\subseteq \Sigma \times\Sigma_i$ to $\mathcal{S}_i$ and, (2) in order to simulate stuttering transitions, adding a component $d_i$ that contains only one location with self-loop transitions labeled with symbols in $\Sigma_{3-i}$ (the local alphabet of the other set of components).
Then the following assume-guarantee rule can be used to decompose the synthesis task into two smaller sub-tasks:

\[
\begin{array}{rcll}
\mathcal{L}(\mathcal{S}_{1+})\cap \mathcal{L}(R) \cap \mathcal{L}(A)  &=& \emptyset &\ \ \ \ \ \ \ \ (a)\\
\mathcal{L}(\mathcal{S}_{2+})\cap \mathcal{L}(\overline{A})&=&\emptyset&\ \ \ \ \ \ \ \ (b)\\
\hline
\mathcal{L}(\mathcal{S}_{+})\cap \mathcal{L}(R) &=& \emptyset &\ \ \ \ \ \ \ \ (c)
\end{array}
\]

The above assume-guarantee rule says that $\mathcal{S}_{+}$ is B-Safe with respect to $\mathcal{L}(R)$ iff there exists an assumption automaton $A$ such that (1) $\mathcal{S}_{1+}$ is B-Safe with respect to $\mathcal{L}(R) \cap \mathcal{L}(A)$ and (2) $\mathcal{S}_{2+}$ is B-Safe with respect to $\mathcal{L}(\overline{A})$, where $\overline{A}$ is the complement of $A$, $\mathcal{P}_{+} = \mathcal{P}_{1} \cup \mathcal{P}_{2}$ and no conflict in $\mathcal{P}_{1}$ and $\mathcal{P}_{2}$. In the following, we prove the above assume-guarantee rule is both sound and complete. Nevertheless, it is unsound for deadlock freeness. An example can be found at the beginning of Section~\ref{sec.algo.prioritysyn.complexities}.

\begin{theo}[Soundness]\label{the:ag-sound}
Let $\mathcal{P}_{1}$ and $\mathcal{P}_{2}$ be two non-conflicting priority rules, $A$ be the assumption automaton, $R$ be the risk specification automaton, $\mathcal{S}_{1+} = (C_1 \cup \{d_1\}, \Sigma, (\mathcal{P}\cap \Sigma\times\Sigma_1) \cup \mathcal{P}_1)$, and $\mathcal{S}_{2+} = (C_2\cup \{d_2\} , \Sigma, (\mathcal{P}\cap \Sigma\times\Sigma_2) \cup \mathcal{P}_2)$, where $\mathcal{P}_i  \subseteq \Sigma \times \Sigma_i$ for $i=1,2$ and $\mathcal{P} \subseteq \Sigma \times (\Sigma_1 \cup \Sigma_2)$. If $\mathcal{L}(\mathcal{S}_{1+})\cap \mathcal{L}(R) \cap \mathcal{L}(A)=\emptyset$ and $\mathcal{L}(\mathcal{S}_{2+})\cap \mathcal{L}(\overline{A})=\emptyset$. The priority rule $\mathcal{P}_{1} \cup \mathcal{P}_{2}$ ensures that the system $\mathcal{S} = (C_1 \cup C_2 , \Sigma, \mathcal{P})$ is B-Safe with respect to $R$.
\end{theo}

\begin{proof}
First, from $\mathcal{L}(\mathcal{S}_{1+})\cap \mathcal{L}(R) \cap \mathcal{L}(A)=\emptyset$ and $\mathcal{L}(\mathcal{S}_{2+})\cap \mathcal{L}(\overline{A})=\emptyset$, we can obtain the relation between those languages described in Figure~\ref{fig:agproof}. From the figure, one can see that the two languages $\mathcal{L}(\mathcal{S}_{1+})\cap \mathcal{L}(R)$ and $\mathcal{L}(\mathcal{S}_{2+})$ are disjoint. This follows that $\mathcal{L}(\mathcal{S}_{1+})\cap \mathcal{L}(R) \cap \mathcal{L}(\mathcal{S}_{2+}) =\emptyset$. By Lemma~\ref{lem:composition}, we have $\mathcal{L}(\mathcal{S}_+)\cap \mathcal{L}(R)\subseteq \mathcal{L}(\mathcal{S}_{1+})\cap \mathcal{L}(\mathcal{S}_{2+})\cap \mathcal{L}(R) =\emptyset$. Hence the set of priorities $\mathcal{P}_{1} \cup \mathcal{P}_{2}$ ensures that $\mathcal{S}$ is B-Safe with respect to $R$.
\end{proof}

\begin{figure}
\centering
 \includegraphics[width=0.45\columnwidth]{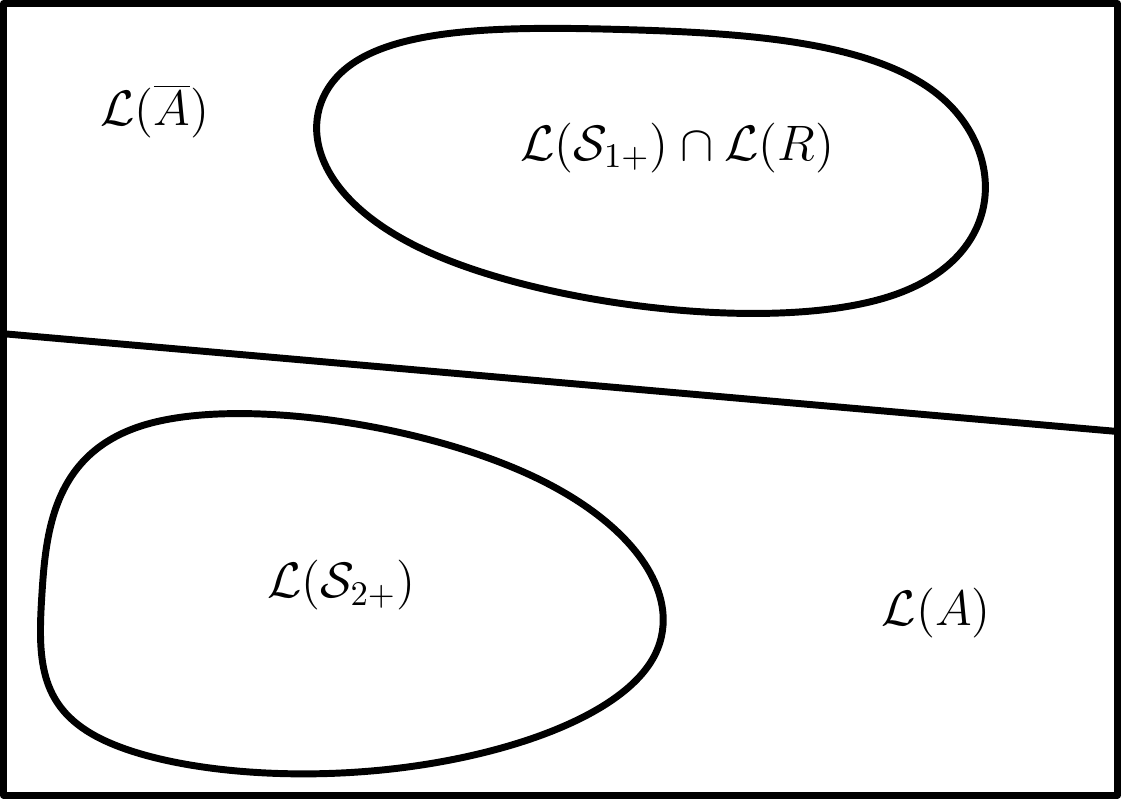}
  \caption{The relation between the languages.}
 \label{fig:agproof}
\end{figure}

\begin{lemma}[Composition]\label{lem:composition}
Let $\mathcal{S}_{1} = (C_1 \cup \{d_1\}, \Sigma, \mathcal{P}_1)$, and $\mathcal{S}_{2} = (C_2 \cup \{d_2\}, \Sigma, \mathcal{P}_2)$, and $\mathcal{S}_{1+2} = (C_1 \cup C_2 , \Sigma, \mathcal{P}_1 \cup \mathcal{P}_2)$
be three systems, where $\mathcal{P}_i  \subseteq \Sigma \times \Sigma_i$ for $i=1,2$. We have $\mathcal{L}(\mathcal{S}_{1+2})\subseteq \mathcal{L}(\mathcal{S}_{1})\cap \mathcal{L}(\mathcal{S}_{2})$.
\end{lemma}
\begin{proof}
For a word $w=\sigma_1, \ldots, \sigma_n \in \mathcal{L}(\mathcal{S}_{1+2})$, we consider inductively from the first interaction.
If $\sigma_1$ is enabled in the initial configuration $(l_1,v_1,\ldots, l_n,v_n, \ldots l_m,v_m)$ of $\mathcal{S}_{1+2}$, then according to Definition~\ref{def:semantics}, we have (1) if $\sigma_1$ is in the interaction alphabet of component $c_i \in C_1 \cup C_2$, then there exist a transition $(l_i,g_i,\sigma_1 , f_i, l_i')$ in $c_i$ such that $g_i(v_i)=\true$ and (2) there exists no transition $(l_i,g_i,\sigma' , f_i, l_i')$ in components of $C_1$ and $C_2$ such that $g_i(v_i)=\true$ and $(\sigma_1,\sigma')\in \mathcal{P}_1 \cup \mathcal{P}_2$.

We want to show that $\sigma_1$ is also enabled in the initial configuration of $\mathcal{S}_{1}$. In order to do this, we have to prove (1) components in
$C_1\cup\{d_1\}$ can move with $\sigma_1$ and (2) there exists no transition $(l_i,g_i,\sigma' , f_i, l_i')$ in $C_1\cup \{d_i\}$ such that $g_i(v_i)=\true$, $l_i$ is an initial location, and $(\sigma_1,\sigma')\in \mathcal{P}_1$.
\begin{itemize}
  \item For (1), we consider the following cases:
        (a) If $\sigma_1 \in \Sigma_{12}$, components of $C_1$ can move with $\sigma_1$ and $d_1$ can move with $\sigma_1$ via a self-loop transition.
        (b) If $\sigma_1 \in \Sigma_{1}$, components of $C_1$ can move with $\sigma_1$ and it is not an interaction of $d_1$.
        (c) If $\sigma_1 \in \Sigma_{2}$, it is not an interaction of $C_1$ and $d_1$ can move with $\sigma_1$ via a self-loop transition.
        Therefore, components in $C_1\cup\{d_1\}$ can move with $\sigma_1$.

  \item For (2), first, it is not possible to have such a transition in any component of
        $C_1$ by the definition of $\mathcal{S}_{1+2}$ and Definition~\ref{def:semantics}. Then, if the transition is in $d_i$, we have $\sigma' \in \Sigma_2$ and it follows that $(\sigma,\sigma') \notin \mathcal{P}_1\subseteq \Sigma\times\Sigma_1$.
\end{itemize}

By the above arguments for (1) and (2), $\sigma_1$ is enabled in the initial configuration of $\mathcal{S}_{1}$. By a similar argument, $\sigma_1$ is also enabled in the initial configuration of $\mathcal{S}_{2}$.

The inductive step can be proved using the same argument. Thus $w \in \mathcal{L}(\mathcal{S}_{1})$ and $w \in \mathcal{L}(\mathcal{S}_{2})$.
It follows that $\mathcal{L}(\mathcal{S}_{1+2})\subseteq \mathcal{L}(\mathcal{S}_{1})\cap \mathcal{L}(\mathcal{S}_{2})$.
\end{proof}

\begin{theo}[Completeness]\label{the:ag-complete}
Let $\mathcal{S}_{+} = (C , \Sigma, \mathcal{P}\cup \mathcal{P}_{+})$ be a system and $R$ be the risk specification automaton. If $\mathcal{L}(\mathcal{S}_{+})\cap \mathcal{L}(R) = \emptyset$, then there exists an assumption automaton $A$, system components $C_1$ and $C_2$ such that $C=C_1\cup C_2$, $C_1\cap C_2 =\emptyset$, and two non-conflicting priority rules $\mathcal{P}_{1}\subseteq \Sigma \times \Sigma_1$ and $\mathcal{P}_{2}\subseteq \Sigma \times \Sigma_2$ such that $\mathcal{L}(C_1 \cup \{d_1\}, \Sigma, \mathcal{P}\cup \mathcal{P}_1)\cap \mathcal{L}(R) \cap \mathcal{L}(A)=\emptyset$,  $\mathcal{L}(C_2 \cup \{d_2\}, \Sigma, \mathcal{P}\cup \mathcal{P}_2)\cap \mathcal{L}(\overline{A})=\emptyset$, and $\mathcal{P}_{+}=\mathcal{P}_{1} \cup \mathcal{P}_{2}$.
\end{theo}
\begin{proof}
Can be proved by taking $C_1=C$, $C_2=\emptyset$, $A$ as an automaton that recognizes $\Sigma^*$, $\mathcal{P}_{1}=\mathcal{P}_{+}$, and $\mathcal{P}_{2}=\emptyset$.
\end{proof}

Below we give an example that if we allow the priority $\mathcal{P}$ to be any relation between the interactions, then the assume-guarantee rule we used is unsound. The key is that Lemma~\ref{lem:composition} will no longer be valid with the relaxed constraints to the priority.
In Figure~\ref{fig:agboundary}, both $C_1$ and $C_2$ has only one components, $\Sigma_1=\emptyset$, $\Sigma_2=\{c\}$, and $\Sigma_{12}=\{a,b\}$.
Assume that we have the priority rule $\mathcal{P}=\{b\prec a\}$ in $\mathcal{S}_1$, $\mathcal{S}_2$, and $\mathcal{S}$.
Then we get $\mathcal{L}(\mathcal{S}_1)=\{a\}$, $\mathcal{L}(\mathcal{S}_2)=\{b + ca\}$, which implies $\mathcal{L}(\mathcal{S}_1)\cap\mathcal{L}(\mathcal{S}_2)=\emptyset$. However, $\mathcal{L}(\mathcal{S})=\{b\}$. Then we found a counterexample for Lemma~\ref{lem:composition}. This produces a counterexample of the soundness of the assume-guarantee rule.
With a risk specification $\mathcal{L}(R)=\{b\}$, an assumption automaton $\mathcal{L}(A)=\Sigma^*$, and priorities $\mathcal{P}=\mathcal{P}_1=\mathcal{P}_2=\{b \prec a\}$,
the subtasks of the assume-guarantee rule can be proved to be B-Safe. However, the system $\mathcal{S}$ is not B-Safe with respect to $\mathcal{L}(R)$. The reason why $\Sigma_{12}$ can not be placed on the right-hand side of $\mathcal{P}$, $\mathcal{P}_1$, and $\mathcal{P}_2$ is because even in the subsystem a shared interaction can block other interactions successfully, when composing two systems together, it may no longer block other interactions (as now they need to be paired).

\begin{figure}
\centering
 \includegraphics[width=0.4\columnwidth]{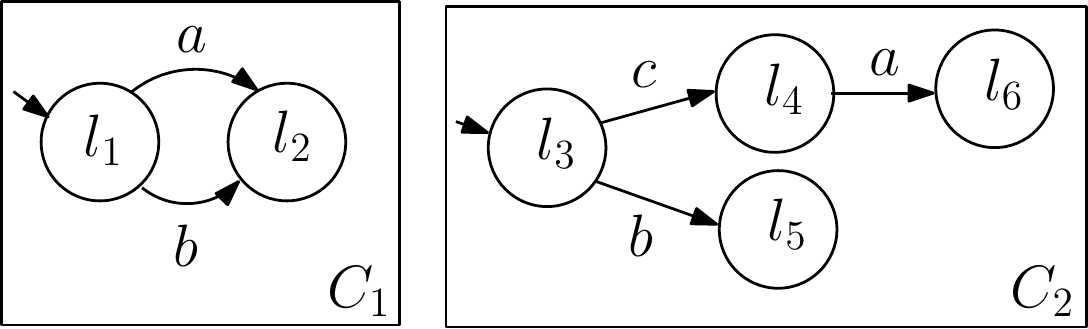}
  \caption{A counterexample when we allow a shared interaction to have higher priority than others.}
 \label{fig:agboundary}
\end{figure}

Notice that (1) the complexity of a synthesis task is NP-complete in the number of states in the risk specification automaton product with the size of the system and (2) $|\mathcal{S}|$ is approximately equals to $|\mathcal{S}_1| \times |\mathcal{S}_2|$\footnote{This is true only if the size of the alphabet is much smaller than the number of reachable configurations.}.
Consider the case that one decomposes the synthesis task of $\mathcal{S}$ with respect to $\mathcal{L}(R)$ into two subtasks using the above assume-guarantee rule. The complexity original synthesis task is NP-complete in $|\mathcal{S}|\times|R|$ and the complexity of the two sub-tasks are $|\mathcal{S}_1|\times|R|\times|A|$ and $|\mathcal{S}_2|\times|A|$\footnote{Since $A$ is deterministic, the sizes of $A$ and its complement $\overline{A}$ are identical.}, respectively.
Therefore, if one managed to find a small assumption automaton $A$ for the assume-guarantee rule, the complexity of synthesis can be greatly reduced. We propose to use the machine learning algorithm L*~\cite{angluin1987learning} to automatically find a small automaton that is suitable for compositional synthesis. Next, we will first briefly describe the L* algorithm and then explain how to use it for compositional synthesis.

The L* algorithm works iteratively to find a minimal deterministic automaton recognizing a target regular language $U$. It assumes a \textit{teacher} that answers two types of queries: (a) \textit{membership queries} on a string $w$,  where the teacher returns \emph{true} if $w$ is in $U$ and \emph{false} otherwise, (b) \textit{equivalence queries} on an automaton $A$, where the teacher returns \emph{true} if $\mathcal{L}(A)=U$, otherwise it returns \emph{false} together with a counterexample string in the difference of $\mathcal{L}(A)$ and $U$.
In the $i$-th iteration of the algorithm, the L* algorithm acquires information of $U$ by posing membership queries and guess a candidate automaton $A_i$. The correctness of the $A_i$ is then verified using an equivalence query. If $A_i$ is not a correct automaton (i.e., $\mathcal{L}(A)\neq U$), the counterexample returned from the teacher will be used to refine the conjecture automaton of the $(i+1)$-th iteration. The learning algorithm is guaranteed to converge to the minimal deterministic finite state automaton of $U$ in a polynomial number of iterations\footnote{In the size of the minimal deterministic finite state automaton of $U$ and the longest counterexample returned from the teacher.}. Also the sizes of conjecture automata increase strictly monotonically with respect to the number of iterations (i.e., $|A_{i+1}|>|A_{i}|$ for all $i>0$).

\begin{figure}[t]
\centering
 \includegraphics[width=0.7\columnwidth]{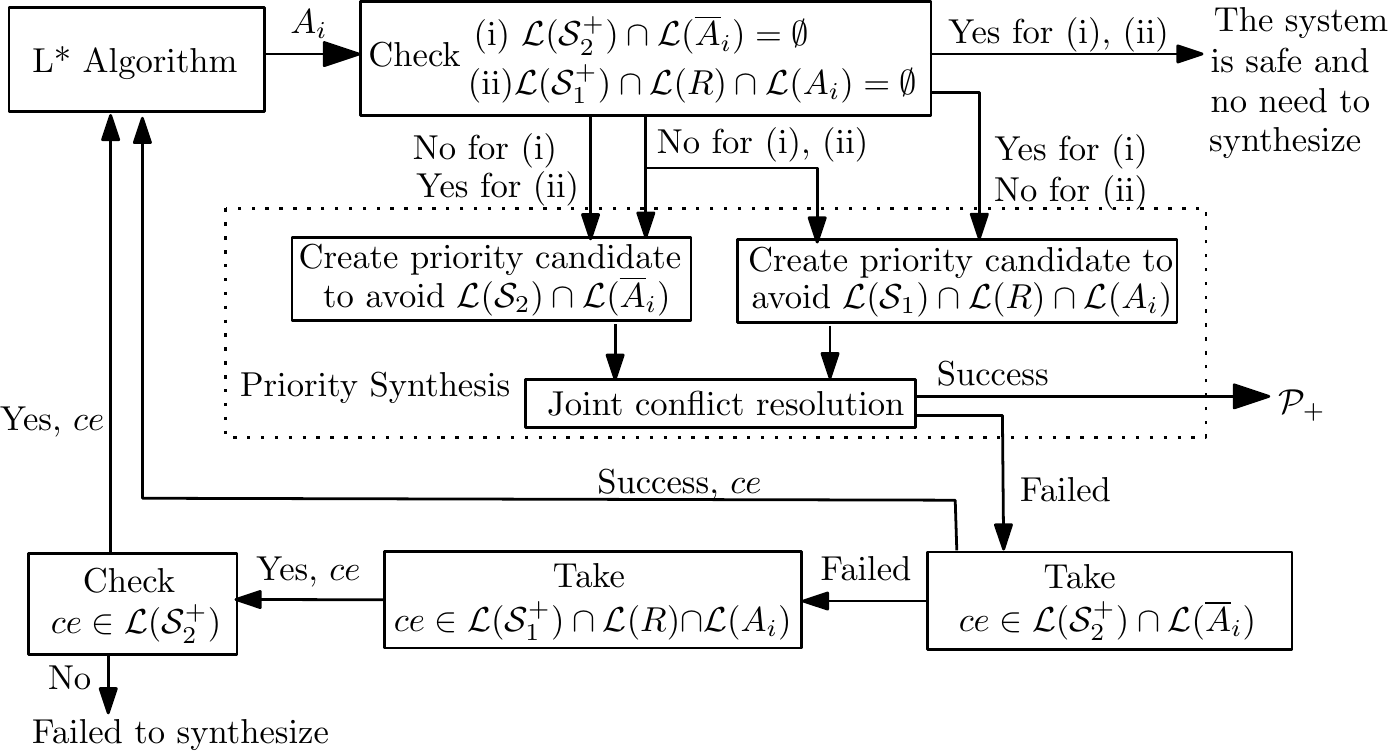}
  \caption{The flow of the assume-guarantee priority synthesis.}
 \label{fig:agflow}
\end{figure}

The flow of our compositional synthesis is in Figure~\ref{fig:agflow}. Our idea of compositional synthesis via learning is the following. We use the notations $\mathcal{S}_{i}^+$ to denote the system $\mathcal{S}_{i}$ equipped with a stuttering component.
First we use L* to learn the language $\mathcal{L}(\mathcal{S}_{2}^+)$. Since the transition system induced from the system $\mathcal{S}_{2}^+$ has finitely many states, one can see that $\mathcal{L}(\mathcal{S}_{2}^+)$ is regular. For a membership query on a word $w$, our algorithm simulates it symbolically on $\mathcal{S}_{2}^+$ to see if it is in $\mathcal{L}(\mathcal{S}_{2}^+)$. Once the L* algorithm poses an equivalence query on a deterministic finite automaton $A_i$, our algorithm tests conditions $\mathcal{L}(\mathcal{S}_{1}^+)\cap \mathcal{L}(R) \cap \mathcal{L}(A_i) = \emptyset$ and $\mathcal{L}(\mathcal{S}_{2}^+)\cap \mathcal{L}(\overline{A_i})=\emptyset$ one after another.
So far, our algorithm looks very similar to the compositional verification algorithm proposed in~\cite{cobleigh2003learning}.
There are a few possible outcomes of the above test
 \begin{enumerate}
   \item Both condition holds and we proved the system is B-Safe with respect to $\mathcal{L}(R)$ and no synthesis is needed.
   \item At least one of the two conditions does not hold. In such case, we try to synthesize priority rules to make the system B-Safe (see the details below).
   \item If the algorithm fails to find usable priority rules, we have two cases:
   \begin{enumerate}
     \item The algorithm obtains a counterexample string $ce$ in $\mathcal{L}(\mathcal{S}_{1}^+)\cap \mathcal{L}(R) \setminus \mathcal{L}(\overline{A_i})$ from the first condition. This case is more complicated. We have to further test if $ce\in \mathcal{L}(\mathcal{S}_{2}^+)$. A negative answer implies that $ce$ is in $\mathcal{L}(A_i) \setminus \mathcal{L}(\mathcal{S}_{2}^+)$. This follows that $ce$ can be used by L* to refine the next conjecture. Otherwise, our algorithm terminates and reports not able to synthesize priority rules.
     \item The algorithm obtains a counterexample string $ce$ in $\mathcal{L}(\mathcal{S}_{2}^+) \setminus \mathcal{L}(A_i)$ from the second condition, in such case, $ce$ can be used by L* to refine the next conjecture.
   \end{enumerate}
 \end{enumerate}

The deterministic finite state automata $R$, $A_i$, and also its complement $\overline{A_i}$ can be treated as components without data and can be easily encoded symbolically using the approach in Section~\ref{subsec.algo.prioritysyn.encoding}. Also the two conditions can be tested using standard symbolic reachability algorithms.

\paragraph{Compositional Synthesis} Recall that our goal is to find a set of suitable priority rules via a small automaton $A_i$. Therefore, before using the $ce$ to refine and obtain the next conjecture $A_{i+1}$, we first attempt to synthesis priority rules using $A_i$ as the assumption automaton.
Synthesis algorithms in previous sections can then be applied separately to the system composed of \{$\mathcal{S}_{1}^+$, $R$, $A_i$\} and the system composed of \{$\mathcal{S}_{2}^+$, $\overline{A_i}$\} to obtain two non-conflicting priority rules $\mathcal{P}_{1i}\subseteq (\Sigma_1\cup\Sigma_{12})\times\Sigma_1$ and $\mathcal{P}_{2i}\subseteq (\Sigma_2\cup\Sigma_{12})\times\Sigma_2$. Then $\mathcal{P}_{1i} \cup \mathcal{P}_{2i}$ is the desired priority for $\mathcal{S}$ to be B-Safe with respect to $R$.
To be more specific, we first compute the CNF formulae $f_1$ and $f_2$ (that encode all possible priority rules that are \emph{local}, i.e., we remove all non-local priority candidates) of the two systems separately using the algorithms in Section~\ref{sec.algo.prioritysyn.repair}, and then check satisfiability of $f_1 \wedge f_2$.
The priority rules $\mathcal{P}_{1i}$ and $\mathcal{P}_{2i}$ can be derived from the satisfying assignment of $f_1 \wedge f_2$.


\section{Evaluation\label{sec.algo.prioritysyn.evaluation}}

\begin{table}[htp]
\center
\begin{threeparttable}
\caption{Experimental results \label{tab:philo}}
\begin{small}
\begin{tabular}{|l|llll|llll|l|}\hline
 & \multicolumn{4}{c|}{Time (seconds)} &\multicolumn{4}{c|}{$\#$ of BDD variables} & \\ \hline
Problem & NFM\tnote{1} & Opt.\tnote{2} &  Ord.\tnote{3} & Abs.\tnote{4} & NFM & Opt. &  Ord. & Abs. & Remark \\ \hline
Phil. 10 &0.813 & 0.303  & 0.291 & 0.169    & 202 & 122 & 122 & 38 & $^1$ Engine based on~\cite{cheng:vissbip:2011} \\
Phil. 20 &-     & 86.646 & 0.755 & 0.166      & -   & 242 & 242 & 38 & $^2$ Dense var. encoding\\
Phil. 25 &-     & -      & 1.407 & 0.183    & -   & -   & 302 & 38 & $^3$ Initial var. ordering\\
Phil. 30 &-     & -      & 3.740  & 0.206    & -   & -   & 362 & 38 & $^4$ Alphabet abstraction\\
Phil. 35 &-     & -      & 5.913 & 0.212    & -   & -   & 422 & 38 & - Timeout/Not evaluated \\
Phil. 40 &-     & -      & 10.210 & 0.228    & -   & -   & 482 & 38 &\\
Phil. 45 &-     & -      & 18.344 & 0.213   & -   & -  & 542 & 38 &\\
Phil. 50 &-     & -      & 30.384 & 0.234   & -   & -  & 602 & 38& \\ \hline\hline
DPU v1  & 5.335   & 0.299  &  x  & x   & 168  & 116  & x & x & $^R$ Priority repushing\\
DPU v2  & 4.174   & 0.537   &  1.134\tnote{R}  & x   & 168  & 116  & 116\tnote{R} & x & x Not evaluated\\ \hline\hline
Traffic  & x  & x   &  0.651  & x   & x  & x & 272 & x &\\ \hline
\end{tabular}
\end{small}
%
\end{threeparttable}
\end{table}

We implemented the presented algorithms (except connection the data abstraction module in D-Finder~\cite{bensalem:dfinder2:2011}) in the \textsc{VissBIP}\footnote{Available for download at \url{http://www6.in.tum.de/~chengch/vissbip}} tool and performed experiments to evaluate them. To observe how our algorithm scales, in Table~\ref{tab:philo} we summarize results of synthesizing priorities for the dining philosophers problem\footnote{Evaluated under Intel 2.93GHz CPU with 2048Mb RAM for JVM.}. Our preliminary result in~\cite{cheng:vissbip:2011} fails to synthesize priorities when the number of philosophers is greater than~$15$ (i.e., a total of $30$ components), while currently we are able to solve problems of $50$ within reasonable time. By analyzing the bottleneck, we found that $50\%$ of the execution time are used to construct clauses for transitive closure, which can be easily parallelized. Also the synthesized result (i) does not starve any philosopher and (ii) ensures that each philosopher only needs to observe his left and right philosopher, making the resulting priority very desirable. Contrarily, it is possible to select a subset of components and ask to synthesize priorities for deadlock freedom using alphabet abstraction. The execution time using alphabet abstraction depends on the number of selected components; in our case we select~$4$ components thus is executed extremely fast. Of course, the synthesized result is not very satisfactory, as it starves certain philosopher. Nevertheless, this is unavoidable when overlooking interactions done by other philosophers.
Except the traditional dining philosophers problem, we have also evaluated on  (i) a BIP model (5 components) for data processing in digital communication (DPU; See Appendix~\ref{appsub.algo.prioritysyn.dpu} for description) (i) a simplified protocol of automatic traffic control (Traffic). Our preliminary evaluation on compositional priority synthesis is in Appendix~\ref{appsub.algo.prioritysyn.compositional.evaluation}.

\section{Related Work\label{sec.algo.prioritysyn.related}}

For deadlock detection, well-known model checking tools such as SPIN~\cite{holzmann:2004:smc} 
and NuSMV~\cite{cimatti1999nns} support deadlock detection by given certain formulas to specify the property.  D-Finder~\cite{bensalem:dfinder2:2011} applies compositional and incremental methods to compute invariants for an over-approximation of reachable states to verify deadlock-freedom automatically. Nevertheless, all the above tools do not provide any deadlock avoidance strategies when real deadlocks are detected.

Synthesizing priorities is subsumed by the framework of controller synthesis proposed by Ramadge and Wohnham~\cite{ramadge1989control}, where the authors proposed an automata-theoretical approach to restrict the behavior of the system (the modeling of environment is also possible).
Essentially, when the environment is modeled, the framework computes the risk attractor and creates a centralized controller.
Similar results using centralized control can be dated back from~\cite{balemi1993supervisory} to the recent work by Autili~et~al~\cite{autili2007synthesis} (the SYNTHESIS tool). Nevertheless, the centralized coordinator forms a major bottleneck for system execution. Transforming a centralized controller to distributed controllers is difficult, as within a centralized controller, the execution of a local interaction of a component might need to consider the configuration of all other components.

Priorities, as they are stateless, can be distributed much easier for performance and concurrency. E.g., the synthesized result of dining philosophers problem indicates that each philosopher only needs to watch his left and right philosophers without considering all others. We can continue with known results from the work of Graf et al.~\cite{GrafPQ10} to distribute priorities, or partition the set of priorities to multiple controllers under layered structure to increase concurrency (see work by Bonakdarpour et al.~\cite{Bonakdarpour2011distribute}). Our algorithm can be viewed as a step forward from centralized controllers to distributed controllers, as architectural constraints (i.e., visibility of other components) can be encoded during the creation of priority candidates. Therefore, we consider the work of Abujarad et al.\cite{abujarad2009parallelizing}  closest to ours, where they proceeds by performing distributed synthesis (known to be undecidable~\cite{PnueliFOCS90})  
directly. In their model, they take into account the environment (which they refer it as faults), and consider handling deadlock states by either adding mechanisms to recover from them or preventing the system to reach it. 
It is difficult to compare two approaches directly, but we give hints concerning performance measure: (i) Our methodology and implementation works on game concept, so the complexity of introducing the environment does not change. (ii) In~\cite{abujarad2009parallelizing}, for a problem of $10^{33}$ states, under $8$-thread parallelization, the total execution time is $3837$~seconds, while resolving the deadlock of the $50$ dining philosophers problem (a problem of $10^{38}$ states) is solved within $31$ seconds 
using our monolithic engine.

Lastly, the research of deadlock detection and mechanisms of deadlock avoidance is an important topic within the community of Petri nets (see survey paper~\cite{li2008survey} for details). Concerning synthesis, some theoretical results are available, e.g.,~\cite{iordache2002synthesis}, but efficient implementation efforts are, to our knowledge, lacking.

\section{Conclusion\label{sec.algo.prioritysyn.conclusion}}

 In this paper, we explain the underlying algorithm for priority synthesis and propose extensions to synthesize priorities for more complex systems. Figure~\ref{fig:algo.prioritysyn.framework} illustrates a potential flow of priority synthesis. A system can be first processed using data abstraction to create models suitable for our analysis framework. Besides the monolithic engine, two complementary techniques are available to further reduce the complexity of problem under analysis. Due to the stateless property and the fact that they preserve deadlock-freedom, priorities can be relatively easily implemented in a distributed setting.

\begin{figure}[t]
\centering
 \includegraphics[width=0.6\columnwidth]{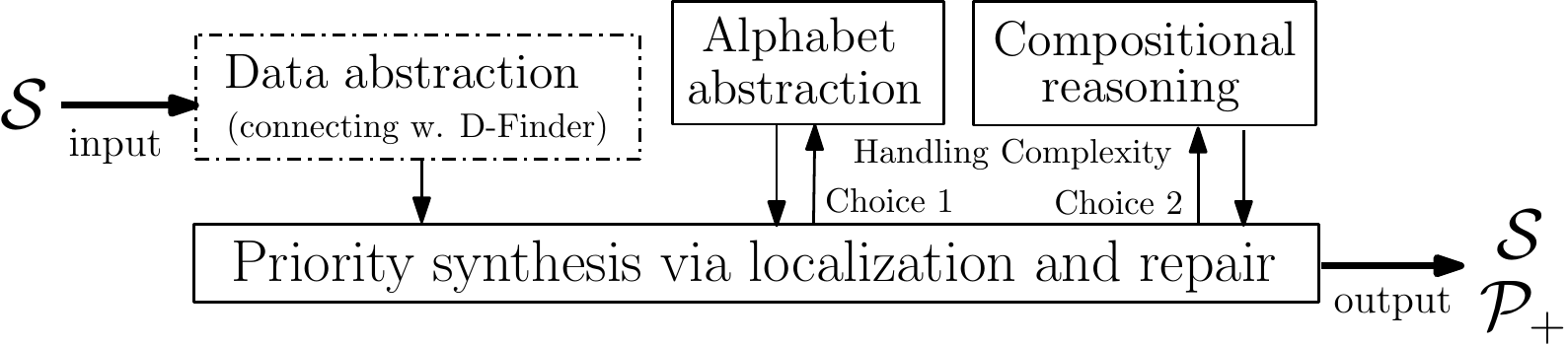}
  \caption{The framework of priority synthesis presented in this paper, where the connection with the D-Finder tool~\cite{bensalem:dfinder2:2011} is left for future work.}
 \label{fig:algo.prioritysyn.framework}
\end{figure}

\appendix
\section{Appendix\label{app.algo.prioritysyn.proofs}}

\subsection{Data Processing Units in Digital Communication\label{appsub.algo.prioritysyn.dpu}}

In digital communication, to increase the reliability of data processing units (DPUs), one common technique is to use multiple data sampling.
We have used \textsc{VissBIP} to model the components and synchronization for a simplified DPU. In the model, two interrupts (\texttt{SynchInt} and \texttt{SerialInt} respectively) are invoked sequentially by a \texttt{Master} to read the data from a \texttt{Sensor}. The \texttt{Master} may miss any of the two interrupts. Therefore, \texttt{SerialInt} records whether the interrupt from \texttt{SynchInt} is lost in the same cycle. If it is missed, \texttt{SerialInt} will assume that the two interrupts have read the same value in the two continuous cycles. According to the values read from the two continuous cycle, \texttt{Master} calculates the result. In case that the interrupt from \texttt{SerialInt} is missing in the second cycle or both interrupts are missing in the first cycle, \texttt{Master} will not calculate anything. Ideally, the calculation result from \texttt{Master} should be the same as what is computed in \texttt{SerialInt}. The mismatch will lead to global deadlocks.

The synthesis of \textsc{VissBIP} focuses on the deadlock-freedom property. First, we have selected the non-optimized engine. \textsc{VissBIP} reports that it fails to generate priority rules to avoid deadlock, in 4.174 seconds with 168 BDD variables. Then we have selected the optimized engine and obtained the same result in 0.537 seconds with 116 BDD variables. The reason of the failure is that two contradictory priority rules are collected in the synthesis. Finally, we have allowed the engine to randomly select a priority between the contradicts (priority-repushing). A successful priority is finally reported in 1.134 seconds to avoid global deadlocks in the DPU case study.

\subsection{Compositional Priority Synthesis: A Preliminary Evaluation\label{appsub.algo.prioritysyn.compositional.evaluation}}

Lastly, we conduct preliminary evaluations on compositional synthesis using dining philosophers problem. Due to our system encoding, when decomposing the philosophers problem to two subproblems of equal size, compare the subproblem to the original problem, the number of BDD variables used in the encoding is only $22.5\%$ less. This is because the saving is only by replacing component construction with the assumption; for interactions, they are all kept in the encoding of the subsystem. Therefore, if the problem size is not big enough, the total execution time for compositional synthesis is not superior than than monolithic method, as the time spent on inappropriate assumptions can be very costly. Still, we envision this methodology more applicable for larger examples, and it should be more applicable when the size of alphabet is small (but with lots of components).

\bibliographystyle{abbrv}

\end{document}